\documentclass[
  10pt, a4paper
]{article}
\usepackage{geometry}
\usepackage{amsmath,amssymb,mathtools,braket}
\usepackage{iftex}
\usepackage[T1]{fontenc}
\usepackage[utf8]{inputenc}
\usepackage{textcomp} 
\usepackage{lmodern}
\usepackage[toc,page]{appendix}

\usepackage[backend=biber, style=numeric, sorting=none]{biblatex}
\numberwithin{equation}{section}

\usepackage{dsfont}

\usepackage{amsthm}
\theoremstyle{definition}
\newtheorem{definition}{Definition}[section]
\theoremstyle{plain}
\newtheorem{theorem}{Theorem}[section]
\theoremstyle{plain}
\newtheorem{lemma}{Lemma}[section]
\theoremstyle{plain}
\newtheorem{proposition}{Proposition}[section]
\theoremstyle{remark}
\newtheorem*{remark}{Remark}

\title{Digitizing lattice gauge theories in the magnetic basis: reducing the breaking of the fundamental commutation relations}
\author{Simone Romiti, Carsten Urbach}
\date{Helmholtz-Institut für Strahlen- und Kernphysik\\%
  and Bethe Center for Theoretical Physics,\\
  University of Bonn, Bonn, Germany\\[2ex]%
  \today
}
\usepackage{hyperref}
\hypersetup{colorlinks=true, allcolors=blue}
\usepackage{cleveref}

\newcommand{\dummy}[1]{#1} 
\newcommand{\Vinv}{V^\dagger} 

\begin{document}
\maketitle
\begin{abstract}
We present a digitization scheme for the lattice $\mathrm{SU}(2)$ gauge
theory Hamiltonian in the \textit{magnetic basis}, where the gauge links are unitary
and diagonal.
The digitization is obtained from a particular partitioning of
the $\mathrm{SU}(2)$ group manifold, 
with the canonical momenta 
constructed by an approximation of the Lie derivatives
on this partitioning.
This construction, analogous to a discrete Fourier transform,
preserves the spectrum of the electric part of
the Hamiltonian and the canonical commutation relations exactly on a subspace of
the truncated Hilbert space, while the residual subspace can be
projected above the cutoff of the theory.
\end{abstract}


\section{Introduction}\label{intro}

A quantum system evolves according to the Schrödinger equation~\footnote{
  Here and throughout the whole paper we will use natural units:
  $\hbar=c=1$.}: \begin{equation}
\label{eq:SchroedingerEquation}
i \frac{\partial}{\partial t} | \psi \rangle = H | \psi \rangle \, ,
\end{equation}
where $H$ is the Hamiltonian operator and
$\ket{\psi}$ lies in the Hilbert space $\mathcal{H}$ of physical
states.
Despite the fact that the solution to Schrödinger's equation is
formally given by the time evolution operator $U_t = \exp(-i H t)$,
finding a solution in practice when no analytical solution is
available is hampered by the so-called \textit{curse of
dimensionality}: the relevant Hilbert space grows in general
exponentially with the system size.

Quantum gauge theories represent such quantum systems for which in general no
analytical solution is known. They are attacked by discretizing space
and restrict to a finite volume, an approach known as lattice gauge
theory. The degrees of freedom are then given by a number of gauge
links per lattice site, with an infinite dimensional Hilbert space
corresponding to a single link $U$. The total Hilbert space is eventually
given by the tensor product of all these local spaces, growing
exponentially with the number of lattice sites.

Even if a digital quantum computer with $n$ qubits was available,
it would allow to implement a Hilbert space of size $2^n$ only, which
is always finite. Therefore, the infinite dimensional site-local
Hilbert space of a single gauge link must be truncated to a finite
dimensional one in order to map it to the space spanned by the
qubits. 
This is even more true if classical computers in combination for
instance with tensor network methods are to be used.
Such a truncation of $\mathcal{H}$ is not unique, and one needs to
check that the correct dynamics is restored when the vector space
becomes
infinite-dimensional~\cite{santhanam1976quantum,STOVICEK1984157}
while optimizing for efficiency.

For lattice gauge theories an important aspect in the truncation of
the Hilbert space is to preserve as much of the gauge symmetry and the
fundamental commutation relations as possible.

The relevant Hamiltonian $H$ for a non-Abelian $\mathrm{SU}(N_c)$ lattice
gauge theory~\cite{PhysRevD.11.395} is
constructed from gauge field operators $U$ and their canonical momenta
$L_a$ and $R_a$ ($a=1,\ldots, N_c^2-1$) at each point $x$ and direction $\mu$ of
the lattice. Local gauge symmetry dictates the form of $H$. While
deferring the exact definitions to Sec.~\eqref{sec-TheorBkg}, the
properties the $U$ and $L_a$, $R_a$ need to fulfil read:
\begin{enumerate}
\item
  canonical commutation relations: $[L_a, U] = -\tau_a U$ and $[R_a, U] = U \tau_a$ 
\item
  structure of the Lie algebra: $[L_a, L_b] = i f_{abc} L_c$ and $[R_a, R_b] = i f_{abc} R_c$,
\item
  special unitarity~\footnote{From now on the label ``special'' will
    always be implied.}: $U\cdot U^\dagger = U^\dagger\cdot U = 1$ and
  $\det(U) = 1$,
\item
  the $L_a$, $R_a$ are ultra-local, i.e.~the only non-vanishing components
  couple neighbor points in the discretized manifold.
\end{enumerate}

Here $\tau_a$ are the generators of $\mathrm{SU}(N_c)$ and $f_{abc}$ are
the $\mathrm{SU}(N_c)$ structure constants. Since $\mathrm{SU}(N_c)$ is parametrized
by continuous parameters, any implementation on a quantum device will,
as mentioned above, require a truncation of the Hilbert space. An
inevitable consequence of such a truncation is that some of the
properties 1.-4. cannot be implemented exactly.

This offers a challenge and an opportunity at the same time. One
question to ask is how to digitize the $L_a$ and $U$ operators such
that efficient practical simulations are feasible. Another question to
ask is which sub-set of the properties 1.-4. can be maximally preserved,
i.e.~an equivalent to the famous Nielsen-Ninomiya theorem for lattice
Dirac operators~\cite{Nielsen:1981hk}. One
could even speculate whether there is an exact lattice version of
local gauge invariance, just like there is an exact version of chiral
symmetry on the lattice~\cite{Ginsparg:1981bj,Luscher:1998pqa}.

There are multiple digitization prescriptions known, see for instance
Refs.~\cite{PhysRevD.101.114502,davoudi2022general,Liu:2021tef,Zache:2023dko,Bauer:2023jvw}. 
Most of them try to preserve the canonical commutation relations. For
instance, quantum link models~\cite{Chandrasekharan:1996ih,Wiese:2021djl}
enjoy exact properties 1., 2. and 4., while the $U$ operators are no
longer unitary but become parts of a larger group. The same is true when
a truncated Clebsch--Gordan expansion is used to represent the $U$
operators (see prop. 3 of Ref.~\cite{PhysRevD.91.054506}).

Only recently we have studied in Ref.~\cite{Jakobs:2023lpp} an approach where
properties 1. and 2. are only approximately fulfilled, while the gauge
field operators remain unitary and the $L_a$,$R_a$ ultra-local. This
approach is based on finite sub-sets of $\mathrm{SU}(N_c)$ elements and
corresponding discretizations of the canonical momenta. In this paper we
will show how, by a clever choice of $\mathrm{SU}(N_c)$ elements in the
sub-set and construction of the momenta, 
properties 2. and 3. are exact and property 1. is exact on a
subspace of the total truncated Hilbert space, but the $L_a$ are no
longer ultra-local.

We will work for the special case of a $\mathrm{SU}(2)$ lattice gauge theory,
which shares many properties with the $\mathrm{SU}(3)$ theory~\cite{PhysRevD.11.395,Sommer_1994},
while significantly simpler to simulate. We also discuss how to
generalize the arguments presented below to $\mathrm{SU}(3)$.

The structure of the paper is the following. In Sec.~\eqref{sec-MainResults} we summarize the main results of this paper,
postponing the details to the subsequent sections. Sec.~\eqref{sec-TheorBkg} recalls some theoretical background, while the
details of our results are discussed in the subsections of
Sec.~\eqref{sec-su2Theory}. Sec.~\eqref{sec-su2Preliminary} shows 
some general arguments on $\mathrm{SU}(2)$ sampling which are called back
later. Sec.~\eqref{sec-DJT} gives the definition and properties of what
we define as ``Discrete Jacobi Transform'', namely the orthogonal
transform based on the Jacobi polynomials that relates the electric and
magnetic bases. In Sec.~\eqref{sec-su2MatrixOperators} we give the
explicit representation of the gauge links and canonical momenta at
finite truncation, together with some discussion on the numerical
implementation and advantage of our digitization. 
App.~\eqref{sec-FiniteDifferences} presents a more expensive approach which
however easily generalizes to $\mathrm{SU}(3)$. 
Finally, in Sec.~\eqref{sec-conclusion} we draw our conclusion and give an outlook for
future developments.

In order to ease the reading we have also included some appendices, with
material that should be fairly standard to the already experts in the
field. In App.~\eqref{sec-SUNcTheoryReview} we recall the main
properties of the $\mathrm{SU}(N_c)$ lattice Hamiltonian degrees of freedom, and
in~\eqref{sec-GeometrysuNc} we discuss the differential geometry
interpretation of the canonical momenta. 

\section{Main results}\label{sec-MainResults}

In this section we summarize the main results of this paper, which we
are going to prove in the following sections. As mentioned already in
the introduction, we have discussed in Ref.~\cite{Jakobs:2023lpp} how to define the
canonical momenta $L_a$, $R_a$ for the case of $\mathrm{SU}(2)$ gauge theories
formulated on a finite partitioning of $S_3$, which is isomorphic to
the group.

With a finite partitioning of $\mathrm{SU}(2)$ we mean a finite set
$\mathcal{D}_m\subset\mathrm{SU}(2)$ of $m$ group elements with the
property that $\mathcal{D}_m$ becomes asymptotically dense in $\mathrm{SU}(2)$
when $m\to\infty$. 
We have given examples for such partitionings in
Refs.~\cite{Jakobs:2023lpp,Hartung:2022hoz}.
These partitionings require one to discretize the differential
operators $L_a$, $R_a$ or $L^2$ directly on these sets.
This however leads to a breaking of the properties 1. and 2. mentioned in the introduction~\eqref{intro},
and hence of gauge invariance.
Both properties are recovered only with $m\to\infty$. On the other hand, 
at finite $m$ the operators
$U$ are unitary and the $L_a$, $R_a$ strictly local,
i.e.~properties 3. and 4. are preserved.

In this paper we are going to show that by giving up the locality
property of the discrete canonical momenta, one can preserve property 1.
exactly on a subspace of the truncated Hilbert space, and maintain the
Lie algebra structure and unitarity of $U$.
This is done by defining a specific
partitioning $\mathcal{D}_{N_\alpha}$ of $S_3$ with
$N_\alpha\in\mathbb{N}$ 
elements based on Euler angles $\vec\alpha$. With this specific
partitioning defined in the following sections there exists a
$q\in\{n/2, n\in\mathbb{N}\}$ and
\begin{equation} \label{eq:NqDefinition}             
  N_q\ =\ \frac{1}{3} (4q +3)(q+1)(2q + 1)\ <\ N_\alpha\,,
\end{equation}
with:
\begin{enumerate}
\item[(a)] the $N_q$ eigenvalues of $L^2$ with main quantum number
  $j\leq q$ reproduced exactly. 
\item[(b)] the Lie algebra structure preserved on the subspace
spanned by the corresponding $N_q$ (discretized) eigenvectors of $L^2$.
\item[(c)] the fundamental commutation relations exactly reproduced on
the subspace spanned by the discretized eigenvectors corresponding to
the smallest $N_{q-1/2}$ eigenvalues of $L^2$.
\item[(d)] the remaining $(N_\alpha - N_{q-1/2})$-dimensional part
  of the truncated Hilbert space possible to be projected to arbitrary
  energies above the cutoff.
\end{enumerate}
For these properties to hold for a given $q$, the partitioning must
have at least 
\begin{equation}
  N_\alpha =
  \begin{cases}
    (q + 1) \cdot (4q + 1) \cdot (4q + 1) \, , \, &\text{if} \, q \in \mathbb{N} \, ,\\
    (q + 1/2) \cdot (4q + 1) \cdot (4q + 1) \, , \, &\text{otherwise} \, ,
  \end{cases}
\end{equation}
elements, as will also be shown below.

From the above list of properties it should be clear that the proposed
scheme is expected to work well if physical states can be approximated
by linear combinations of the aforementioned $N_{q-1/2}$ eigenvectors
of $L^2$. This expectation is, however, not that far-fetched and
underlies in fact all the schemes working in a basis where $L^2$ is
diagonal, truncated at some maximal eigenvalue. So, certainly at large
values of the coupling the proposed scheme is expected to work with not
too large values of $N_q$.

In summary, our approach defines a $\mathrm{SU}(2)$ effective Hamiltonian,
preserving the unitarity of the links. The Lie algebra, equations of
motion (and Gauss' law) are fulfilled below the cutoff. Unitary links
allow for a gauge invariant state preparation~\cite{PhysRevD.11.395}, with the $U$s being
implementable as gates on a quantum device. The direct discretization of
the manifold allows for a direct comparison with Lagrangian simulations,
and for instance compare to Ref.~\cite{Hartung:2022hoz}.

An implementation for the construction of the operators described in this work can be found in \cite{DJTPaperRepo}.

\section{Theoretical background}
\label{sec-TheorBkg}

\subsection{Remarks on the digitization of $\mathrm{SU}(2)$}
\label{sec-RemarksSU2Digitization}

In this work we focus on the digitization of the gauge links and the
momenta needed to simulate the standard Wilson lattice Hamiltonian with
$N_c=2$~\cite{PhysRevD.15.1128}:
\begin{equation}
  \label{eq:KGHamiltonian}
  H = 
  \frac{g^{2}}{4}
  \sum_{\vec{x}} \sum_{\mu=1}^{d-1} \sum_{a=1}^{3} 
  \left[{(L_a)}_{\mu}^{2}(\vec{x})+{(R_a)}_{\mu}^{2}(\vec{x}) \right]-
  \frac{1}{g^{2}}
  \sum_{\vec{x}} \sum_{\mu=1, \nu<\mu}^{d-1} \,
  \mathrm{Tr}[{U}_{\mu \nu}(\vec{x}) + {U}_{\mu \nu}^\dagger(\vec{x}) ] \, .
\end{equation}
The total Hilbert space is given by the tensor product of the spaces corresponding to each pair $(\vec{x}, \mu)$. 
Therefore, in order to digitize $H$ it is sufficient
to address the problem of digitizing the degrees of freedom
``pointwise'', i.e.~at each point $\vec{x}$ and direction $\mu$. In the
following, when writing $U$, $L_a$, $R_a$, it will be understood
that these correspond to a pair $(\vec{x}, \mu)$. These have to fulfil the
relations (see App.~\eqref{sec-SUNcTheoryReview} for a review on these
properties for $\mathrm{SU}(N_c)$): \begin{equation}
\label{eq:CanonicalCommRel}
[L_a, U] = - \tau_a U \, ,\quad
[R_a, U] = U \tau_a \, .
\end{equation}
\begin{equation}
\label{eq:LaRbContraints}
 [L_a, R_b] = 0 \, , \quad
 \sum_a L_a L_a = \sum_a R_a R_a \, ,
\end{equation}
\begin{equation}
\label{eq:su2LaRaCommRelLieAlgebra}
 [L_a, L_b] = i \epsilon_{abc} L_c \, , \quad
 [R_a, R_b] = i \epsilon_{abc} R_c \, .
\end{equation}
The relations in Eq.~\eqref{eq:su2LaRaCommRelLieAlgebra} are the
standard commutation relations of quantum angular momentum~\cite{sakurai_napolitano_2017}. Using the
constraint of Eq.~\eqref{eq:LaRbContraints} we find that the irreps are
labeled by $3$ half-integer quantum numbers $(j, m_L, m_R)$:
\begin{equation}
\label{eq:jmLmRIrrepsSU2}
| j, m_L, m_R \rangle \, , \,\, 2j \in \mathbb{N} \, , \, |m_L|,|m_R| \leq j \, .
\end{equation}
This is the \textit{electric} basis, and the generators act
on its elements as follows: 
\begin{align}
\label{eq:su2IrrepsOperators}
\left(\sum_a R_a^2\right) | j, m_L, m_R \rangle = 
\left(\sum_a L_a^2\right) | j, m_L, m_R \rangle = 
j(j+1) | j, m_L, m_R \rangle \, , \\
\label{eq:su2L3convention}
L_3 | j, m_L, m_R \rangle =  m_L | j, m_L, m_R \rangle \, , \\
\label{eq:su2R3convention}
R_3 | j, m_L, m_R \rangle = -m_R | j, m_L, m_R \rangle \, , \\
\label{eq:su2Lpmconvention}
(L_1 \pm i L_2) | j, m_L, m_R \rangle = \sqrt{j(j+1) - m_L (m_L \pm 1)}| j, m_L \pm 1, m_R \rangle \, , \\
\label{eq:su2Rpmconvention}
(R_1 \mp i R_2) | j, m_L, m_R \rangle = -\sqrt{j(j+1) - m_R (m_R \pm 1)}| j, m_L, m_R \pm 1 \rangle \, .
\end{align}
$m_L$($m_R$) is the \textit{left}(\textit{right}) magnetic
quantum number, with degeneracy $(2j+1)$ at fixed $j$ and
$m_R$($m_L$). Therefore, the degeneracy of the main quantum number $j$ is
$(2j + 1)^2$. The vacuum $|0\rangle$ of the electric Hamiltonian is
the $j=0$ state. In fact (recall that $|m_L|,|m_R| \leq 0$):
\begin{align}
L_3 |0\rangle = R_3 |0\rangle = 0 |0\rangle = \vec{0} \, , \\
L_\pm |0\rangle = R_\pm |0\rangle = \vec{0} \, , 
\end{align} where $\vec{0}$ is the null vector of the Hilbert space.
In other words,
${L_a |0\rangle = R_a |0\rangle = \vec{0} \, , \, \forall a}$. 


In an infinite dimensional Hilbert space we can write an exact solution
for $U$~\cite{PhysRevD.91.054506,davoudi2022general}.
As already found in Ref.~\cite{PhysRevD.11.395},
$U |0 \rangle \sim |1/2\rangle$, where the magnetic quantum numbers
$m_L$ and $m_R$ are specified by the choice of the matrix element of
$U$ in color space. In general, $U|j\rangle$ will be a combination
of the $|j+1/2\rangle$ and $|j-1/2\rangle$ states (cf.~Eq.~(27) of Ref.~\cite{davoudi2022general}). In a finite
dimensional space however, this ladder-like behavior of $U$ cannot
continue indefinitely, resulting in a ``boundary effect'' on the space
of truncated irreps of the $su(2)$ algebra.

We conclude with the following remark. A basis for Gauss' law invariant
states can be obtained by applying gauge invariant operators to the
vacuum $|0\rangle$~\cite{PhysRevD.11.395}.
Having unitary links in a quantum simulation allows an initial state
preparation that is automatically invariant, using the links as quantum
gates. This avoids the presence on unphysical contributions that need to
be removed by enforcing Gauss' law a posteriori with, e.g., a penalty
term. 

\subsection{Asymptotic behavior}\label{asymptotic-behavior}

In the following we use a basis
$\{ | U \rangle\}$ of the group elements eigenstates. 
Namely we work with operators that are functionals in the space of the wavefunctions
$\psi(U)$ which are $L^2$-integrable with respect to the \textit{Haar} measure~\cite{PhysRevD.15.1128}. 
Since the manifold $S_3$ is isomorphic to the group
$\mathrm{SU}(2)$, we can use the manifold points $p$ to
label the elements of the group: $|p\rangle \equiv |U(p)\rangle$ and
$\psi(p) \equiv \psi(U(p))$.

In this formalism, the gauge links $U$ are $\mathrm{SU}(2)$ matrices in the
fundamental representation, while the momenta $L_a$, $R_a$ are
differential operators.
For instance,
the $L_a$ are represented by (cf.~e.g.~Refs.~\cite{Shnir:2005xx,murata2008separability}):
\begin{align} 
\label{eq:L1.diff.op.continuum}
\dummy{L}_1 &= 
- i \left(
- \sin{\phi} \frac{\partial}{\partial \theta} 
- \cos{\phi} \cot{\theta} \frac{\partial}{\partial \phi} 
+ \frac{\cos{\phi} }{\sin{\theta}} \frac{\partial}{\partial \psi} \right)
 \, , 
\\
\label{eq:L2.diff.op.continuum}
\dummy{L}_2 &= 
- i \left(
+ \cos{\phi} \frac{\partial}{\partial \theta}
- \sin{\phi} \cot{\theta} \frac{\partial}{\partial \phi} 
+ \frac{\sin{\phi}}{\sin{\theta}} \frac{\partial}{\partial \psi} 
\right)
\, , 
\\
\label{eq:L3.diff.op.continuum}
\dummy{L}_3 &= - i \frac{\partial}{\partial \phi} 
\, ,
\end{align} where $\theta$, $\phi$, $\psi$ are the Euler angles
charting $S_3$. The operator $\sum_a \dummy{L}_a \dummy{L}_a$
is, therefore, given by: 
\begin{equation} \label{eq:Lsquared.diff.op.continuum}
\sum_a \dummy{L}_a \dummy{L}_a =
- \cot{\theta} \frac{\partial}{\partial \theta}  
- \frac{\partial^{2}}{\partial \theta^{2}} 
- \frac{1}{\sin^2{(\theta)}} \frac{\partial^{2}}{\partial \phi^{2}}  
+ 2 \frac{\cos{\theta}}{\sin^2{\theta}} \frac{\partial^{2}}{\partial \psi\partial \phi} 
- \frac{1}{\sin^{2}{\theta}}  \frac{\partial^{2}}{\partial \psi^{2}} 
\end{equation}
When truncating the Hilbert space to a dimension $N$, the states
become vectors in a finite-dimensional vector space, and the operators
endomorphisms on the latter. This inevitably results in some
approximation of the commutation relations. In fact, while it is in
principle possible to preserve the Lie algebra 
(see Eq.~\eqref{eq:su2LaRaCommRelLieAlgebra}), 
the canonical commutation relations
of Eq.~\eqref{eq:CanonicalCommRel} are to be understood in the
distributional sense. By taking the Hilbert space trace on the
left and right-hand sides, we see that in a finite-dimensional vector
space they never hold with unitary links (cf.~\cite{de2003quantum}). Therefore, numerically,
the infinite truncation limit has to be verified as a convergence of
their action on arbitrary functions $\Phi$~\cite{Jakobs:2023lpp,garofalo2022defining}:
\begin{equation}
\label{eq:LUCommutatorDistributionalSense}
\left( [L_a, U] + \tau_a U \right)  \vec{\Phi} \to 0 \, ,
\end{equation}
and analogously for the $R_a$. Given a functional
$\Phi(U)$, $\vec{\Phi}$ is a (normalized) vector whose components
converge to the values of $\Phi$ when $N \to \infty$.

We conclude this section by observing the following.
Discretizing a continuous group manifold in general breaks
gauge invariance as, e.g., the multiplication of $2$ elements of the group may not lie in the set~\cite{hamermesh2012group}. 
This can be a non-negligible problem at the renormalization level~\cite{Lepage:1998dt}.
In the Hamiltonian formulation this can be solved by a penalty term~\cite{PhysRevLett.125.030503},
leading to an effective gauge-invariant theory below the cutoff.
In our prescription, the degrees of freedom behave like the continuum manifold ones on a subspace of the truncated Hilbert space, and the projector to the remaining subspace can be used to build such a penalty term in the Hamiltonian.

\section{$\mathrm{SU}$$(2)$ theory and construction of the momenta}\label{sec-su2Theory}

In this section we provide the finite-dimensional representations of the
canonical momenta. Sec.~\eqref{sec-su2Preliminary} discusses some
general arguments about $S_3$ which are used in the subsequent
sections. We then provide an explicit construction sharing the
aforementioned properties (a)-(d) of Sec.~\eqref{sec-MainResults}.
Note that there is an alternative construction based on finite
difference operators, which we summarise for completeness in
App.~\eqref{sec-FiniteDifferences}. 

\subsection{Frequencies on $S_3$}
\label{sec-su2Preliminary}

As mentioned previously, we are working in a basis of eigenstates of
the (unitary) operator $U$, discretized by means of a partitioning
$\mathcal{D}_{N_\alpha}$. For this we have to express the operators
$L^2$ and $L_a$ ($R^2$, $R_a$) in this basis while trying to preserve
as much of the continuum properties of the operators mentioned in the
introduction~\eqref{intro} as possible. The handle we can use to
optimise for the continuum properties is the choice of the element of
$\mathcal{D}_{N_\alpha}$ as well as the construction of the momenta. 
More specifically, the strategy is to 
choose these elements such that the lowest $N_q$ (see Eq.~\eqref{eq:NqDefinition}) continuum eigenvectors
of $L^2$ (and $R^2$) can be uniquely represented on this partitioning.

In the continuum manifold limit, we can use the isomorphy
of $\mathrm{SU}(2)$ to $S_3$~\cite{Shnir:2005xx}, labelling the
elements of the partitioning by the three Euler
angles $\theta, \phi, \psi$ with the following convention: 
\begin{equation}
\label{eq:EulerAnglesBounds}
0 \leq \theta \leq \pi\,, \quad 0 \leq \phi \leq 4\pi\,, \quad 0 \leq \psi \leq 4\pi\,.
\end{equation}
In the continuum manifold, the irreducible
representations of the $su(2)$ algebra are labelled by $3$
half-integers $(j, m_L, m_R)$, where $2j \in \mathbb{N}$ and
$m_L, m_R=-j,\ldots,j$, and the eigenfunctions of $L^2$ are the
Wigner functions $D^j_{m_L m_R}$~\cite{Shnir:2005xx}:
\begin{equation}
\label{eq:WignerDFunctionsDefinition}
D^{j}_{m_L m_R}(\vec{\alpha}) \equiv D_{m_L m_R}^{j}(\phi,\theta,\psi)
= e^{i m_L \phi}d_{m_L m_R}^{j}(\theta)e^{i m_R \psi} \, ,
\end{equation}
with $d_{m_L m_R}^j$ the so-called Wigner d-functions~\cite{Celeghini:2014uxa}
\begin{equation}
  \label{eq:WignerdFunctionsDef}
  d_{m_L m_R}^{j}(\theta)=
  \left(\frac{(j-m_L)!(j+m_L)!}{(j-m_R)!(j+m_R)!}\right)^{\frac{1}{2}}(1-x)^{\frac{m_L
  + m_R}{2}}(1+x)^{-\frac{m_L-m_R}{2}} J^{m_R - m_L, m_R + m_L}_{j -
  m_R} (x) \, , \quad x = \cos{(\theta)} \, , 
\end{equation}
where $J$ is a Jacobi polynomial.
One can interpret $(j,m_L,m_R)$ as the Fourier
frequencies for $\theta, \phi, \psi$ respectively, 
as functions on $S_3$ can be spectrally decomposed
by means of the $D^J_{m_L m_R}$ 
(see Sec.~4.10 of Ref.~\cite{doi:10.1142/0270}):
\begin{equation}
\label{eq:S3SpectralDecompositionContinuum}
f(\vec{\alpha}) = f(\phi,\theta,\psi) = \sum_{j=0}^{\infty} \sum_{m_L, m_R = -j}^{j} a^j_{m_L m_R} D_{m_L m_R}^{j}(\phi,\theta,\psi) \, .
\end{equation}
This is analogous to the spherical harmonics transform on the sphere $S_2$~\cite{sansone1959orthogonal}.

According to the Nyquist-Shannon (or Whittaker-Kotelnikow-Shannon) sampling theorem~\cite{1697831},
the sampling rate for each direction must be at least twice the bandwidth to
be able to uniquely reconstruct functions up to the corresponding maximal frequency.
Thus, on a finite partitioning, if $N_q$ is the number of modes with $j\leq q$, in general we will need
a significantly larger $N_\alpha$ to be able to represent these $N_q$
modes exactly on $\mathcal{D}_{N_\alpha}$.
This property is the analogous of the sampling theorems for $S_2$~\cite{6006544}.
From the physical point of
view we are saying that frequencies higher than a threshold $j=q$
have to be treated as unphysical. However, since we regularize the
gauge theory with a cutoff, one can chose $q$ large enough, such that
the modes with $j>q$ are above this cutoff.

\subsection{A Discrete Jacobi Transform on $S_3$}
\label{sec-DJT}

In this section we discuss how to chose the elements of
$\mathcal{D}_{N_\alpha}$ in order to be able to represent the modes
with $j\leq q$ exactly.
The construction is
based on Jacobi polynomials, and can be viewed as a generalization of
the discrete orthogonal transform for $S_2$ found in Ref.~\cite{BARONE199029}. We will use it to build
a finite-dimensional representation of the canonical momenta in Sec.~\eqref{sec-su2MatrixOperators}. 
We start with defining partitionings with the help of polynomials:

\begin{definition}\label{def-PPCircle}

\textbf{(Polynomial partitioning of the circle)}: Consider a set of
orthogonal polynomials
$\{p_k(x) \, , \, k = 0, \ldots, n , \, x=\cos{(\theta)}\}$ such that
$p_n$ has $n$ roots (in $[-1, 1]$). We call the set of these
roots, $\{\theta_s \, , \, s = 1, \ldots, n\}$, a polynomial
partitioning of the circle.

An example of such a partitioning many are familiar with is the one
induced by  Chebyshev polynomials of $1$st kind: 
$p_k(x) = \cos{(k\theta)}$. The roots of $p_N$ are
$\theta_s = \frac{2\pi}{(2j+1)} s , \, s=-j,\ldots,j \,, N=2j+1$.

\end{definition}

\begin{definition}\label{def-PPS3}

\textbf{(Polynomial partitioning of the 3-sphere)}: Let $\{\theta_s\}$
be a polynomial partitioning from def.~\eqref{def-PPCircle}. A
polynomial partitioning of $S_3$ is a set of angular coordinates
$\vec{\alpha}_k=(\theta_a, \phi_b, \psi_c)$ such that:

\begin{align}
  \label{eq:S3PolynomialPartitioning}
\theta_a & \, , \, a = 1, \ldots, N_\theta \\
\phi_b &= \frac{4\pi}{N_\phi} b \, , \, b = 1, \ldots, N_\phi \, , \\
\psi_c &= \frac{4\pi}{N_\psi} c \, , \, c = 1, \ldots, N_\psi \, .
\end{align}

$k=1,\ldots,N_\alpha$, where $N_\alpha=N_\theta N_\phi N_\psi$ is
the total number of points on the sphere.

\end{definition}

For instance, the Legendre-partitioning will be such that the
$\theta_s$ are the roots of the $N_\theta$-th Legendre polynomial,
while the $\phi$ an $\psi$ will always be evenly distributed along
the corresponding circles.

Next we recall the following property of orthogonal polynomials
(cf.~theorem (3.6.12) of Ref.~\cite{king1982introduction}):

\begin{theorem}
\label{thm-PolynomialWeights}

Let $\langle \cdot, \cdot \rangle$ be a scalar product on the linear
space $L^2[a,b]$:
\begin{equation}
  \label{eq:DefInnerProduct}
  \langle f, g \rangle = \int_{a}^{b} \mathrm{d}x \, \omega(x) f(x) g(x)
  \, ,
\end{equation}
where $\omega(x)$ is the weight function. Now let
$\{p_k(x)\}_{k=0,\ldots,n}$ be a set of orthogonal polynomials, and
$x_1,\ldots,x_n$ the roots of $p_n(x)$. If the \textit{weights}
$w_1,\ldots,w_n$ are the solution of the (non-singular) system of
equations:
\begin{equation}
    \label{eq:GaussianWeightsLinearSystem}
    \sum_{s=1}^{n} p_k(x_s) w_s = 
      \begin{cases}
      \langle p_0, p_0 \rangle & k=0\,, \\
      0 & k = 1, \ldots, n-1\,.
      \end{cases}
  \end{equation}
Then $w_s > 0 \, , \, s=1,\ldots,n$ and:
\begin{equation}
  \label{eq:WeightedSumPolynomialEqualIntegral}
  \sum_{s=1}^{n} p(x_s) w_s = 
  \int_{a}^{b} \mathrm{d}x \, \omega(x) p(x) \, ,
\end{equation}
where $p(x)$ is any polynomial of degree less than $2n+1$.

\end{theorem}

The weights $w_s$ are often called Gaussian weights, since they can be
used to integrate numerically $e^{-x^2}$ by Taylor expanding it to a
polynomial of finite degree.

Using theorem~\eqref{thm-PolynomialWeights} we now formulate the
following lemma:

\begin{lemma}
\label{lem-LemmaDJT}

Let $d^{j}_{m_L m_R}(\theta)$ be the Wigner d-functions
introduced in Eq.~\eqref{eq:WignerdFunctionsDef}.
We assume
$2j \in \mathbb{N}, \, j \leq q$ and $|m_L|,|m_R| < j$. Let also
$w_s$ be the weights of theorem~\eqref{thm-PolynomialWeights} with
$n > 2q$ and weight function $\omega(x)=1$. Then if $j_1$ and
$j_2$ are both integers or both half integers:

\begin{equation}
  \label{eq:DJTLemmaStatement}
  \sum_{s=1}^{n} 
    w_s \, d^{j_1}_{m_L m_R}(\theta_s) d^{j_2}_{m_L m_R}(\theta_s) =
    \frac{1}{{j_1 + 1/2}} \delta_{j_1 j_2}
    \, .
\end{equation}

\end{lemma}

\begin{proof}
We recall that under a sign change of $m_L$ and $m_R$ we have
$d^{j}_{m_L, m_R}(\theta) = (-1)^{m_L + m_R} d^{j}_{-m_L, -m_R}(\theta)$ (see~\cite{Celeghini:2014uxa}). Therefore after
some agebra we get,
\begin{equation}
  \label{eq:dFunctionsPropertyWeights}
  \begin{split}
    \sum_{s} w_s \,
    d^{j_1}_{m_L, m_R}(\theta_s) d^{j_2}_{m_L m_R}(\theta_s) &=
    (-1)^{m_L + m_R} \, \sum_{s} 
    w_s d^{j_1}_{-m_L, -m_R}(\theta_s) d^{j_2}_{m_L m_R}(\theta_s) \\
    &= (-1)^{m_L + m_R} \, 
    \sum_{s} w_s \,
    J^{-m_R + m_L, -m_R - m_L}_{j_1+m_R}(x_s) 
    J^{m_R -m_L, m_R+m_L}_{j_2-m_R}(x_s)
    \, .
  \end{split}
\end{equation}
Now, whether $j$ is integer or half-integer,
$j \pm m_R$ is always integer valued. Second, the product of two polynomials
of degree $n_1$ and $n_2$ is $n_1+n_2$, therefore the product
\[
J^{-m_R + m_L, -m_R - m_L}_{j_1+ m_R}(x_s)\,\cdot\, J^{ m_R -m_L,
m_R+m_L}_{j_2- m_R}(x_s)
\]
is a polynomial of degree $j_1+j_2$, which is also integer valued.
Finally, since $j_1 + j_2 \leq 2q < n$ by definition, we can replace
the weighted sum with the integral:
\begin{equation}
  \label{eq:DJTLemmaStep1}
  \begin{split}
    (-1)^{m_L + m_R} \sum_{x_s} w_s \,
    J^{-m_R + m_L, -m_R - m_L}_{j_1+m_R}(x_s) 
    J^{m_R -m_L, m_R+m_L}_{j_2-m_R}(x_s) \\
    =
    (-1)^{m_L + m_R} \int_{-1}^{1} \mathrm{d}x J^{-m_R + m_L, -m_R - m_L}_{j_1+m_R}(x) J^{m_R -m_L, m_R+m_L}_{j_2-m_R}(x) \, .
  \end{split}
\end{equation}
We now use again the $m_L, m_R$ sign change property of the d-functions,
in the other direction, to conclude the proof:
\begin{equation}
  \label{eq:DJTLemmaStep2}
  \begin{split}
    (-1)^{m_L + m_R} \, &\int_{-1}^{1} \mathrm{d}x \,
    J^{-m_R + m_L, -m_R - m_L}_{j_1+m_R}(x) 
    J^{m_R -m_L, m_R+m_L}_{j_2-m_R}(x) \\=
    & \int_{-1}^{1} \mathrm{d}x \,
    d^{j_1}_{m_L, m_R}(\arccos{x}) d^{j_2}_{m_L, m_R}(\arccos{x})
    = \frac{1}{{j_1 + 1/2}} \delta_{j_1 j_2} \, .
  \end{split}
\end{equation}
In the last step we have used the well known orthogonality
property~\cite{Celeghini:2014uxa} 
\begin{equation}
  \label{eq:AJFOrthogonality}
  \int_{-1}^{1} \mathrm{d}x \, 
  \mathcal{J}_{j_1}^{m_L, m_R}(x) \, 
          {(j_1 + 1/2)} \, 
          \mathcal{J}_{j_2}^{m_L, m_R}(x) 
          = \delta_{j_1 j_2}
          \, .
\end{equation}
of the Algebraic Jacobi Polynomials
\[
\mathcal{J}^{m_R,
  -m_L}_{j}(\cos{(\theta)}) = d^{j_1}_{m_L,
  m_R}(\theta)\,.
\]
\end{proof}
Now we can define the discrete transform anticipated above:

\begin{definition} \label{def-DJTS3Sphere}

\textbf{(Discrete Jacobi transform)}: Let $\vec{\alpha}_k$ be a
Polynomial Jacobi-partitioning of $S_3$. The following equation defines
the Discrete Jacobi Transform (DJT) of a function $f$ on on $S_3$:
\begin{equation}
  \label{eq:DJTMatrix}
  f(\vec{\alpha}_k) = f(\theta, \phi, \psi) = \sum_{j=0}^{q} \sum_{m_L, m_R = -j}^{j} {(\operatorname{DJT})}^j_{m_L, m_R}(\vec{\alpha}_k) \hat{f}(j, m_L, m_R)
  \, ,
\end{equation}
where:
\begin{equation}
  {(\operatorname{DJT})}^j_{m_L, m_R}(\vec{\alpha}_k) = 
  (j + 1/2)^{1/2}  
  \sqrt{\frac{w_s}{N_\phi N_\psi}}
  D^j_{m_L, m_R}(\vec{\alpha}_k)
  \, .
\end{equation}
\end{definition}
\noindent We recall that $D^j_{m_L, m_R}$ are the Wigner D-functions of Eq.~\eqref{eq:WignerDFunctionsDefinition}. If we list 
all the values of $f$ on $S_3$ in a vector of size $N_\alpha$, and
all the moments of the distribution $\hat{f}$ in a vector of size
\begin{equation}
  \label{eq:NqSumDegeneracy}
  N_{q}=\sum_{j=0}^{q} (2j+1)^2 = \frac{1}{6} (4q+3)(2q+2)(2q+1)\,,
\end{equation}
Eq.~\eqref{eq:DJTMatrix} can be understood in the matrix sense:
\begin{equation}
  \label{eq:DJTMatrixSense}
  \vec{f}_i = {(\operatorname{DJT})}^{i}_{b} \vec{\hat{f}}_b \, ,
\end{equation}
where the indices $i$ and $b$ are the checkerboard
indices of the $\vec{\alpha}$\,s and the $(j,m_L,m_R)$ triplets,
respectively. One important property of the $\operatorname{DJT}$ is the following:

\begin{theorem}
\label{thm-DJTUnitarity}

If $N_\theta > q$, $N_\phi > 4q$ and $N_\psi > 4q$, the DJT matrix
has orthonormal columns, i.e.:
\begin{equation}
  \label{eq:DJTUnitarity}
        {(\operatorname{DJT})}^\dagger {(\operatorname{DJT})} = \mathrm{1}_{N_q \times N_q}
\end{equation}

\end{theorem}

\begin{proof}
  We need to show that
  $({(\operatorname{DJT})}^\dagger {(\operatorname{DJT})})_{b_1 b_2} = \delta_{b_1 b_2} = \delta_{j_1 j_2} \delta_{{m_L}_1 {m_L}_2} \delta_{{m_R}_1 {m_R}_2}$.
  From the explicit expression of ${(\operatorname{DJT})}$ we find:
  \begin{equation}
    \label{eq:DJTUnitarityProofSteps}
    \begin{split}
      ({(\operatorname{DJT})}^\dagger {(\operatorname{DJT})})_{b_1 b_2} &= {{(\operatorname{DJT})}^{*}}_{b_1}^{k} {(\operatorname{DJT})}^{k}_{b_2} =
      \sum_{k} (v^{j_1}_{{m_L}_1 {m_R}_1})^* (D^{j_1}_{{m_L}_1 {m_R}_1}(\vec{\alpha}_k))^*
      v^{j_2}_{{m_L}_2 {m_R}_2} D^{j_2}_{{m_L}_2 {m_R}_2}(\vec{\alpha}_k) \\
      &=
      \sum_{\theta, \phi,\psi} 
          [(j_1 + 1/2)(j_2 + 1/2)]^{1/2}
          \frac{w_s}{N_\phi N_\psi}
          e^{i({m_L}_2-{m_L}_1)\phi} e^{i({m_R}_2-{m_R}_1)\psi}\times\\ 
          &\qquad\qquad\times d^{j_1}_{{m_L}_1 {m_R}_1}(\theta)
          d^{j_2}_{{m_L}_2 {m_R}_2}(\theta) \\
          &= 
          \delta_{{m_L}_1 {m_L}_2} \delta_{{m_R}_1 {m_R}_2} 
                [(j_1 + 1/2)(j_2+1/2)]^{1/2}
                \sum_{s} w_s d^{j_1}_{{m_L}_1 {m_R}_1}(\theta_s)
                d^{j_2}_{{m_L}_1 {m_R}_1}(\theta_s) \\
                &= 
                \delta_{{m_L}_1 {m_L}_2} \delta_{{m_R}_1 {m_R}_2} \frac{[(j_1 + 1/2)(j_2+1/2)]^{1/2}}{{j_1 + 1/2}}
                \delta_{j_1 j_2} \\
                &=
                \delta_{{m_L}_1 {m_L}_2} \delta_{{m_R}_1 {m_R}_2} \delta_{j_1 j_2}
                \, .
    \end{split}
  \end{equation}
  In the intermediate step we have used the known relation (${N,k \in \mathbb{N},\, N>k}$): 
  \begin{equation}
  \label{eq:DeltaFromSumPhases}
  \sum_{\ell=0}^{N-1} e^{i \frac{2\pi}{N} \ell k} = 
  N \delta_{k, 0} \, ,
  \end{equation}
  which applies here by the definition of the paritioning and the lower
  bounds on $N_\phi$ and $N_\psi$. The $\delta_{{m_L}_1,{m_L}_2}$ and
  $\delta_{{m_R}_1 {m_R}_2}$ ensure also that $j_1$ and $j_2$ are either
  both integers or both half-integers. The final step is done using the
  lemma~\eqref{lem-LemmaDJT}.
\end{proof}

A possible choice (and also the one we will use from now on) is,
therefore:
\begin{equation}
  \label{eq:Nalpha.from.q}
  N_\alpha = N_\theta N_\phi N_\psi =
  \begin{cases}
    (q + 1/2) \cdot (4q + 1) \cdot (4q + 1) \, , \, &\text{if} \, q \notin \mathbb{N} \wedge 2q \in \mathbb{N} \, ,
    \\
    (q + 1) \cdot (4q + 1) \cdot (4q + 1) \, , \, &\text{if} \, q \in \mathbb{N} \, .
  \end{cases}
\end{equation}
%

\subsection{Canonical momenta with the DJT}
\label{sec-su2MatrixOperators}

For $\mathrm{SU}(2)$, classically each gauge link is parametrized by $2$
complex numbers $a, b$:
\begin{equation}
  \mathcal{U} = 
  \begin{pmatrix}
    a & b \\
    -b^* & a^*
  \end{pmatrix}
  \, , \, |a|^2 + |b|^2 = 1 \, .
\end{equation}
This requires one real number and $2$ phases. We can
use the three Euler angles of Eq.~\eqref{eq:EulerAnglesBounds}:
\begin{equation}
  \label{eq:su2UalphaMatrix}
  \begin{split}
    \mathcal{U}(\vec{\alpha}) &= 
    e^{-i \phi \tau_3} e^{-i \theta \tau_2} e^{-i \psi \tau_3}
    =
    \begin{pmatrix}
      \cos({\theta}/{2}) e^{-i(\phi+\psi)/2} & -\sin({\theta}/{2}) e^{-i(\phi-\psi)/2} \\
      \sin({\theta}/{2}) e^{i(\phi-\psi)/2} & \cos({\theta}/{2}) e^{i(\phi+\psi)/2}
    \end{pmatrix}
    \\
    &=
    \begin{pmatrix}
      D^{1/2}_{-1/2, -1/2}(\vec{\alpha}) &
      -D^{1/2}_{-1/2, +1/2}(\vec{\alpha}) \\
      -D^{1/2}_{+1/2, -1/2}(\vec{\alpha}) &
      D^{1/2}_{+1/2, +1/2}(\vec{\alpha})
    \end{pmatrix}
    \, .
  \end{split}
\end{equation}
$\tau_a = \sigma_a/2$ are the the generators of $\mathrm{SU}(2)$ in the fundamental irrep $j=1/2$ (see App.~\eqref{sec-SUNcTheoryReview}),
and $D^j_{m_L, m_R}$ are the Wigner D-functions of Eq.~\eqref{eq:WignerDFunctionsDefinition}.
The generalization of Eq.~\eqref{eq:su2UalphaMatrix} to higher irreps is obtained by replacing the $\tau_a$ with the generators ${(T_j)_a}$ of the $j$-th irrep (cf.~e.g.~\cite{PhysRevD.11.395}).
$\mathcal{U}$ is invariant under the simultaneous transformation ${\phi \to \phi+2\pi}$ and
${\psi \to \psi+2\pi}$. Therefore we need $0 \leq \phi \leq 2\pi$
(see e.g.~appendix A of Ref.~\cite{Shnir:2005xx})
in order to avoid a double counting of the elements. However, even if in
our discretization we have chosen $0 \leq \phi \leq 4\pi$, the double
counting does not happen. In fact, $N_\phi$ and $N_\psi$ are both
odd numbers (see Eq.~\eqref{eq:Nalpha.from.q}), and the above
transformation is never realized in the partitioning.

When we quantize the theory, the eigenstates of the quantum operator
${U}$ are such that: \begin{equation}
\label{eq:UketalphaAction}
{U} | \vec{\alpha} \rangle  = \mathcal{U}(\vec{\alpha}) 
| \vec{\alpha} \rangle
\, .
\end{equation}
%
This means that we can work in
the basis where the links are unitary and diagonal in the Hilbert space:
\begin{equation}
\label{eq:su2UalphaOperatorMatrix}
{U} = 
\begin{pmatrix}
    U_{-1/2, -1/2} &
    U_{-1/2, 1/2} \\
    U_{1/2, -1/2} &
    U_{+1/2, +1/2}
\end{pmatrix}
=
\sum_{k=1}^{N_\alpha} 
|\vec{\alpha}\rangle\,\mathcal{U}(\vec{\alpha})\,\langle\vec{\alpha}|
\, \dot{=} \,
\begin{bmatrix}
\mathcal{U}(\vec{\alpha}_1) & 0 & \cdots & 0 \\
0 & \mathcal{U}(\vec{\alpha}_2) & \cdots & 0 \\ 
\vdots & \vdots & \ddots & \vdots \\
0 & \cdots & \cdots & \mathcal{U}(\vec{\alpha}_{N_{\alpha}})  \\
\end{bmatrix}
\, .
\end{equation}
Using the definition~\eqref{def-DJTS3Sphere} of the discrete Jacobi
transform in matrix form with the notation $V = \operatorname{DJT}$
and the property proven in theorem~\eqref{thm-DJTUnitarity},
we can now give our representations of the truncated momentum
operators in the magnetic basis:
\begin{align}
  \label{eq:LaSimilarityTransform}
  L_a = V \hat{L}_a \Vinv \, , \\
  \label{eq:RaSimilarityTransform}
  R_a = V \hat{R}_a \Vinv \, .
\end{align}
$\hat{L}_a$, $\hat{R}_a$ are the matrix representations
of the generators in the electric basis truncated at $j=q$.
They have the following properties, which we anticipated in the introduction:

\begin{proposition} \label{prp-LiealgebraPreserved}
  The truncated momentum operators in the magnetic basis preserve the
  group algebra:
  \begin{equation}
    \label{eq:LieAlgebraCommeRelWanted}
          [L_a, L_b] = i f_{abc} L_c
  \end{equation}
\end{proposition}

\begin{proof}
  This can be proven using $\Vinv V = \mathrm{1}$:
  \begin{equation}
    \begin{split}
      \label{eq:LieAlgebraCommeRelAfterSimilarityTransform}
            [L_a, L_b] &= L_a L_b - L_b L_a =  
            V \hat{L}_a \Vinv V \hat{L}_b V - 
            V \hat{L}_b \Vinv V \hat{L}_a V \\
            &=
            V \hat{L}_a \hat{L}_b \Vinv - 
            V \hat{L}_b \hat{L}_a \Vinv =
            V [\hat{L}_b, \hat{L}_a] \Vinv \\ 
            &=
            i f_{abc} V \hat{L}_c \Vinv =
            i f_{abc} L_c
    \end{split}
  \end{equation}
  The proof for the $R_a$ is identical.
\end{proof}

\begin{proposition} \label{prp-FirstNqEigenvals}

The first $N_q$ eigenvalues of $\sum L_a^2$ are reproduced exactly,
while the remaining ${N_r = N_\alpha - N_{q}}$ dimensional
subspace belongs to the kernel of $\sum L_a^2$.
\end{proposition}

\begin{proof}
This is true because the columns of $V$ are eigenvectors of
$\sum L_a^2=\sum R_a^2$, $L_3$ and $R_3$.\\
Without loss of generality, we prove this only for $\sum L_a^2$ and
$L_3$ since for the $R_a$ the steps are identical. If $a$ is the
checkerboard index of $(j_a, {m_L}_a, {m_R}_a)$, the vector
$\vec{v}^{j_a}_{{m_L}_a {m_R}_a} = \vec{v}_a$ with components
$V^i_a\, , \, i=1,\ldots,N_\alpha$ satisfies:
\begin{align}
  \label{eq:LSquaredvEigenstateProof}
  \begin{split}
    \left[ \left(\sum_b L_b^2\right) \vec{v}_k\right]^i 
    &=
    \sum_{a_1, a_2, k}
        {V}^{i}_{a_1} 
        \left(\sum_b \hat{L}_b^2 \right)_{a_1 a_2} 
             {\Vinv}^{k}_{a_2} {V}^{k}_a=   
             \sum_{a_1, a_2} 
             V^{i}_{a_1} 
             j_{a_1} (j_{a_1} + 1) \delta_{a_1 a_2}
             \delta_{a_2 a} \\
             &=
             j_a(j_a + 1) [\vec{v}_k]^{i} 
  \end{split}
\end{align} and
\begin{equation}
  \label{eq:L3vEigenstateProof}
  \left[ L_3 \vec{v}_k\right]^i =
  \sum_{a_1, a_2, k}
      {V}^{i}_{a_1} 
      (L_3)_{a_1 a_2} 
      {\Vinv}^{k}_{a_2} {V}^{k}_a =  
      \sum_{a_1, a_2} 
      V^{i}_{a_1} 
      m_{a_1} \delta_{a_1 a_2} \delta_{a_2 a} =
      {m_L}_a [\vec{v}_k]^{i}
\end{equation}
Finally, from the rank-nullity theorem~\cite{Alama+2008+137+142} for $\Vinv$,
there exist $N_r$ states ${\{|r_k\rangle\}_{k=1,\ldots,N_r}}$ such
that $\Vinv |r_k \rangle = \vec{0}$. The $|r_k \rangle$ are the
$N_r$ Wigner functions with $j > q$, and satisfy:
\begin{equation}
  \label{eq:LaRaTrivialAboveq}
  L_a \vec{v}^{j}_{m_L, m_R} = R_a \vec{v}^{j}_{m_L, m_R}  
  = \vec{0} \, , \qquad j > q \, .
\end{equation}
\end{proof}
We remark that the residual $N_r$ states behave like the vacuum, but are not the same
as $|0\rangle$, since the action of $U$ will not necessarily mix
them with $j=1/2$ only.
Fig.~\eqref{fig:spectrum} shows explicitly how our implementation reproduces exactly the $N_q$ eigenvalues of the $L_a$ and $\sum_a L_a L_a$.
\begin{figure}
  \includegraphics[width=0.5\textwidth]{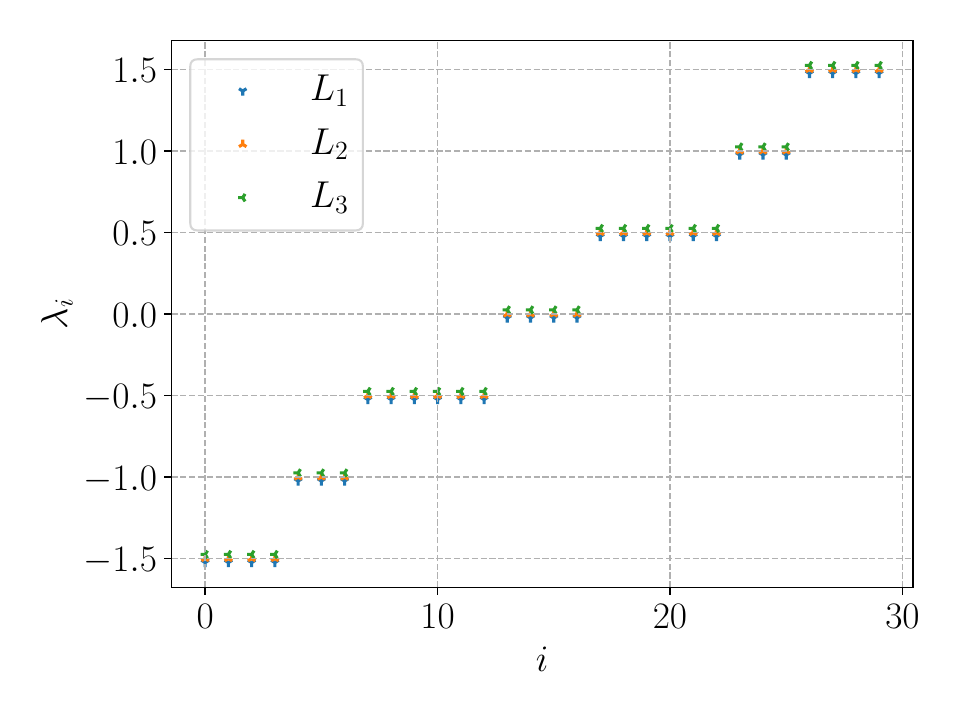}
\includegraphics[width=0.5\textwidth]{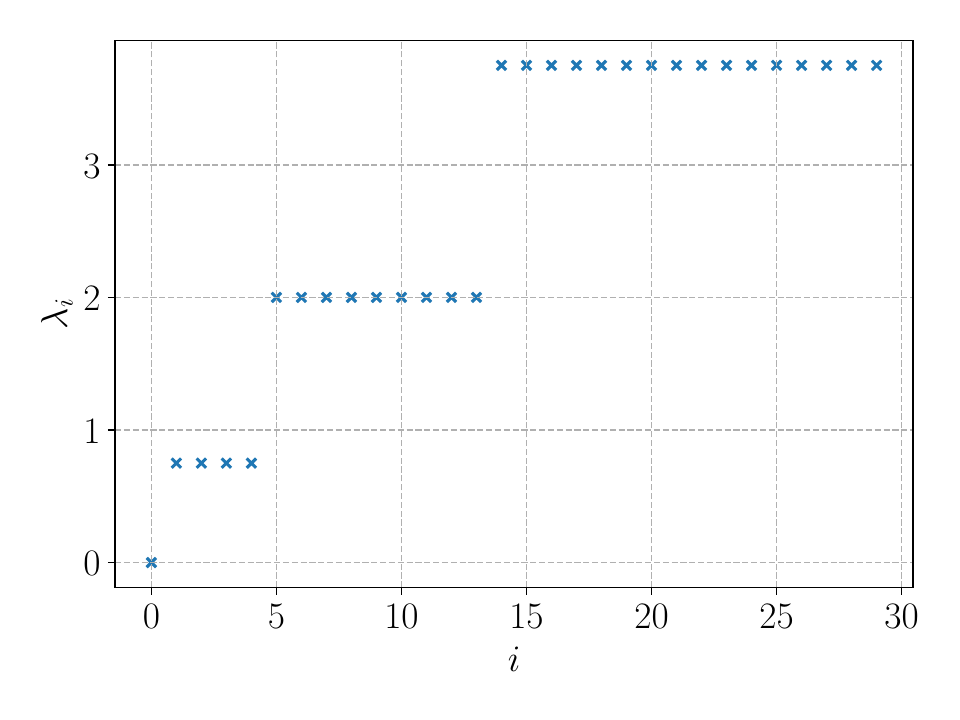}
\caption{
  Numerical spectrum of the canonical momenta obtained with the DJT,
  limited to the first $N_q$ eigenstates with $j \leq q$. 
  For purely illustrative purposes we consider $q=3/2$.
  The horizontal axes are label indices, while the eigenvalues $\lambda_i$ are on the vertical axes.
  Left panel: spectrum of the $L_a$. The eigenvalues are slightly shifted as a function of $a$ for better visualization,
  but are equal to $m_L = -j,\ldots,j$ (for each $j$) up to machine precision.
  Right panel: spectrum of $\sum_a L_a L_a$. 
  We reproduce exactly the eigenvalues $j(j+1)$, for $j=0,\ldots,q$.
  The data have been produced with the implementation of Ref.~\cite{DJTPaperRepo}.
  }
\label{fig:spectrum}
\end{figure}

Let us now discuss the eigenstates of the discretised operators. 

\begin{proposition}
\label{prp-FirstNqEigenstates}

Consider  the the first (linearly independent) $N_q$ eigenstates of $\sum_a L_a L_a$, $L_3$, $R_3$. 
As $q \to \infty$, they approach the naive discretization of the eigenfunctions of the continuum manifold operators,
namely the Wigner $D$-functions.

\end{proposition}

\begin{proof}

The columns of $V$ are not only some eigenstates with
the correct eigenvalues (with Eq.~\eqref{eq:LaSimilarityTransform} this
would be the case for any invertible matrix $V$), but from lemma
(\ref{lem-wsDistribution}) we know that in the $q \to \infty$ limit
their components are the values of the Wigner $D$-functions stacked
into a vector of $N_\alpha$ components. From the continuum manifold
formalism we also know that the electric field operators are represented
by differential operators~\cite{Shnir:2005xx,Tsuchiya:2020hur}
and the eigenfunctions corresponding to the states $|j, m_L, m_R\rangle$
of \Crefrange{eq:su2IrrepsOperators}{eq:su2Rpmconvention} are indeed the Wigner
$D$-functions:
\begin{equation}
  \label{eq:wavefunctionsAreDfunctions}
  \langle \vec{\alpha} | j, m_L, m_R \rangle\ =\
  \langle \theta, \phi, \psi | j, m_L, m_R \rangle\ \propto\
  D^{j}_{m_L m_R}(\theta, \phi, \psi) \, .
\end{equation}
This implies that, up to a normalization factor, for $q \to \infty$ the
columns of $V$ become the continuous manifold eigenfunctions sampled at the points of the $S_3$ partitioning.

The linear independence (at any $q$) follows from the fact that they are eigenvectors
with distinct eigenvalues.

\end{proof}
Finally, we are able to also show that the canonical commutation
relations are exactly reproduced on a subspace of the discretised
Hilbert space. 

\begin{proposition}
\label{prp-FirstNqprimeCommRel}

The canonical commutation relations are reproduced exactly on the first
${N_{q'} = N_{q-1/2}}$ eigenvectors of $\sum L_a^2=\sum_a R_a^2$,
$L_3$ and $R_3$.

\end{proposition}
The proof follows immediately from
the lemmas~\eqref{lem-FirstNqprimeAsContinuum} and prop.~\eqref{prp-CommRelFromDiffOps},
by replacing $\psi(U)$ by an element $\vec{\Phi}$ in the space spanned by the first
$N_{q'}$ vectors $\vec{w}^j_{m_L m_R}$ of lemma~\eqref{lem-FirstNqprimeAsContinuum}.

\noindent Let us make the following remarks:
\begin{enumerate}
\def\labelenumi{\arabic{enumi}.}
\item
  The states with $q < j \leq q'$ have the correct eigenvalues but
  don't fullfill the canonical commutation relations.
\item
  We can always project the residual $N_r = N_\alpha-N_q$ (or
  $N_{r'} = N_\alpha - N_{q'}$) states to whatever energy
  above the cutoff of the theory. If $P_g$ is
  the projector to this ``garbage space'', this is implemented as
  ${L_a \to L_a + \kappa P_g}$, where: 
  \begin{equation}
    \label{eq:ProjectorOutGarbageSpace}
    P_g = \sum_{j > q_t} 
    | j, m_L, m_R \rangle \langle j, m_L, m_R | 
    \, \dot{=} \,
    \sum_{j > q_t} 
    \vec{v}^{j}_{m_L, m_R} \left( {\vec{v}}^{j}_{m_L, m_R} \right)^\dagger 
    \, ,
    \end{equation}
    for some $\kappa \gg 1$ and target truncation $q_t$. 
\item
  One can show (e.g.~by induction) that $N_r = N_\alpha-N_q$ is always
  even. For, if $P_g$ projects to the first $N_r$ states, we can
  also preserve the Lie algebra while projecting above the cutoff:
  \begin{equation}
    L_a \to L_a + 
    P_g \left[ 
   \mathds{1}_{N_q \times N_q} \oplus
   (\tau_a \otimes \mathds{1}_{\frac{N_r}{2} \times \frac{N_r}{2}})
   \right] P_g.
    \end{equation}
\end{enumerate}

\section{Conclusion and outlook}\label{sec-conclusion}

In this paper we have discussed a specific approach to the digitsation
of the $\mathrm{SU}(2)$ lattice gauge theory Hamiltonian,
which is needed for
tensor network or quantum computer based simulations of lattice gauge
theories. 
The digitisation scheme is formulated in a so-called
magnetic basis, where the gauge field operator is diagonal and
unitary, while gauge symmetry is preserved exactly on a subspace of
the truncated Hilbert space. This comes at the price of a dense matrix
representation for the canonical momentum operators.

The approach is based on specific partitioning of the sphere
$S_3\cong\mathrm{SU}(2)$ and a discrete Jacobi transform with the main
property that the $N_q$ continuum eigenfunctions with main quantum
number $j \leq q$ of the electric part $L^2$ in the Hamiltonian can be
exactly and uniquely represented. The remainig states can be shifted
above an energy cutoff and interpreted as integrated out.

It remains to be seen whether this formulation performs more
efficiently than other formulations on the market, for instance
compared to the proposal from Ref.~\cite{Zache:2023dko}. Most importantly, it
needs to be investigated in how far the residual gauge symmetry breaking
spoils simulation results and renormalisability.

Including fermionic fields in this formulation is unproblematic.
First numerical results have been presented at Lattice 2023 for a
$2$-sites $1+1$ dimensional $\mathrm{SU}(2)$ Schwinger type
model~\cite{romitiLattice2023,garofaloLattice2023,Garofalo:2023zkd}, 
for which the discretization introduced in the previous sections
reproduces the spectum exactly. 

\section*{Acknowledgements}\label{acknowledgements}

We warmly thank A. Crippa, M. Garofalo, T. Hartung, T. Jakobs, K. Jansen, J. Ostmeyer, D. Rolfes and U.J. Wiese for the very
interesting and fruitful discussions on this project.
This work is supported by the Deutsche
Forschungsgemeinschaft (DFG, German Research Foundation) and the  
NSFC through the funds provided to the Sino-German
Collaborative Research Center CRC 110 “Symmetries
and the Emergence of Structure in QCD” (DFG Project-ID 196253076 -
TRR 110, NSFC Grant No.~12070131001).

\begin{appendices}

  \section{$\mathrm{SU}(N_c)$ gauge theories on the lattice}
  \label{sec-SUNcTheoryReview}

  In this section we compile some formulae and derivations that are valid for any $\mathrm{SU}(N_c)$ theory.
  The Hamiltonian formulation of lattice gauge theories is known since
  1975~\cite{PhysRevD.11.395}. The starting
  point is the classical action $S$ on a lattice 
  $\Lambda$ with $d$ dimensions~\cite{gattringer2009quantum}:
  \begin{equation}
    \label{eq:GaugeAction}
    S = \frac{2}{g^2} \sum_x \sum_{\mu=0, \nu<\mu}^{d-1} \operatorname{Re}\operatorname{Tr}[ \mathds{1} - \mathcal{U}_{\mu\nu}(x)]
    \, ,
  \end{equation}
  where $x = (t, \vec{x}) \in \Lambda$, 
  the gauge links $\mathcal{U}$ are $\mathrm{SU}(N_c)$ color matrices in the fundamental representation, 
  and $\mathcal{U}_{\mu \nu}(x)$ is the square plaquette:
  \begin{equation}
  \label{eq:PlaquetteFromGaugeLinks}
  \mathcal{U}_{\mu \nu}(x) = \mathcal{U}_{\mu}(x)\,\mathcal{U}_{\nu}(x+{\mu})\,\mathcal{U}_{\mu}^{\dagger}(x+{\nu})\,\mathcal{U}_{\nu}^{\dagger}(x)\,.
  \end{equation}
  The trace is taken in color space. 
  The links are related to the classical ``gluon'' fields
  $A_\mu^a$ of the continuum by~\cite{chin1985exact}: 
  \begin{equation}
    \label{eq:UmuFromAmu}
    \mathcal{U}_\mu(x) = \exp{(\pm i \, a \, \tau_b A_{\mu}^b(x))} \, , b=1,\ldots,N_c^2 - 1 \, ,
  \end{equation}
  where $a$ is the lattice spacing. The choice of the
  $\pm$ is irrelevant, since it fixes a convention on what we call
  $\mathcal{U}_\mu(x)$ or $\mathcal{U}^\dagger_\mu(x)$. 
  $\tau_b$ are the generators of $su(N_c)$ in the fundamental
  representation, which we normalize as $\operatorname{Tr}\{\tau_a, \tau_b\} = \delta_{ab}$~\footnote{
    For $N_c=2$, $\tau_b = \sigma_b/2$, where
    $\sigma_b$ are the Pauli matrices.
  }.
  Under a gauge transformation $V$
  at every point $x$ the links transform as: 
  \begin{equation}
    \label{eq:UgaugeTransformation}
    \mathcal{U}_\mu(x) \to V(x) \mathcal{U}_\mu(x) V^\dagger(x + \mu) \, ,
  \end{equation}
  leaving the action invariant.
  \\
  \\
  We now quantize the theory.
  For every point $x$ and direction $\mu$, 
  the classical degrees of freedom become operators on the Hilbert space $\mathcal{H}$ 
  generated by all possible configurations of gauge links 
  (cf.~e.g.~\cite{PhysRevD.15.1128,gattringer2009quantum}):
  \begin{equation}
    \label{eq:UBasisDef}
  \ket{\mathcal{U}} =  \bigotimes_{x \in \Lambda} \bigotimes_{\mu=0}^{d-1} \ket{\mathcal{U}_\mu(x)} \, .
  \end{equation}
  The link operators $U_\mu(x)$ are defined as:
  \begin{equation}
    U_\mu(x)  = \int \mathrm{d} \mathcal{U}_\mu(x) \,\, \mathcal{U}_\mu(x) \ket{\mathcal{U}_\mu(x)} \bra{\mathcal{U}_\mu(x)} \, ,
  \end{equation}
  where $\mathrm{d} \mathcal{U}$ is the Haar measure of the group~\cite{PhysRevD.91.054506}.
  The presence of a symmetry implies the existence of hermitian generators~\cite{sakurai1995modern}.
  As pointed out in Ref.~\cite{PhysRevD.11.395}, in
  general $V(x)$ and $V^\dagger(x + {\mu})$ in Eq.~\eqref{eq:UgaugeTransformation} are independent of each other.
  Thus, after quantization, the theory has to be invariant under simultaneous Left and Right
  transformations.
  The latter are generated by the hermitian, traceless~\footnote{A $\mathrm{SU}(N_c)$ transformation has unit determinant.} operators $L_a$ and $R_a$, satisfying:
  \begin{equation}
    \label{eq:LandRCommute}
    [L_a, R_b] = 0 \, , \, \forall a,b = 1,\ldots,N_c^2 -1 \, .
  \end{equation}
  The transformations are defined as:
  \begin{align}
    \label{eq:DLdefinition}
  D_{L_\mu(x)}[\Omega] \ket{\mathcal{U}_\mu(x)} &\vcentcolon= \ket{\Omega \, \mathcal{U}_\mu(x)} \, ,\\
    \label{eq:DRdefinition}
  D_{R_\mu(x)}[\Omega] \ket{\mathcal{U}_\mu(x)} &\vcentcolon= \ket{\mathcal{U}_\mu(x) \, \Omega^\dagger} \, ,
  \end{align}
  where $\Omega = e^{i \omega_a \tau_a}$ and, e.g., $D_{L_\mu(x)} = e^{i \omega_a (L_a)_\mu(x)}$.
  As generator of the representations, the $L_a$ and $R_a$ form an $\mathrm{su}(N_c)$ algebra 
  (with structure constants $f_{abc}$):
  \begin{align}
    \label{eq:LaLieAlgebraCommRel}
    [L_a, L_b] &= i f_{abc} L_c \, , \\
    \label{eq:RaLieAlgebraCommRel}
    [R_a, R_b] &= i f_{abc} R_c \, .
  \end{align}
  Furthermore, 
  using Eqs.~\eqref{eq:DLdefinition},~\eqref{eq:DRdefinition} for infinitesimal transformations, 
  we find the canonical commutation relations: 
  \begin{align}
    \label{eq:LaCanonicalCommRel}
    [L_a, U] &= - \tau_a U  \, , \\
    \label{eq:RaCanonicalCommRel}
    [R_a, U] &= U \tau_a \, .
  \end{align}
  In fact, for instance:
  \begin{equation}
    D_L[\Omega] U D_L^\dagger[\Omega] = 
    \int  \mathrm{d} \mathcal{U} \,\, \mathcal{U} \ket{\Omega \, \mathcal{U}} \bra{\Omega \, \mathcal{U}} =
    \Omega^{\dagger} U \, ,
  \end{equation}
  where we have used the invariance of the Haar measure~\cite{gattringer2009quantum}.
  
  In summary, the basis of Eq.~\eqref{eq:UBasisDef} is redundant for the description of the physical states,
  which are equivalent up to arbitrary products of the following local tranformations:
  \begin{equation}
    \label{eq:LocalGaussTransformations}
    D[\Omega]({x}) = e^{i \omega_a G_a(x)} = \prod_{\mu=0}^{d} D_{L_\mu({x})}[\Omega] D_{R_\mu({x}-\mu)}[\Omega] \, .
  \end{equation}
  The generators read: 
  \begin{equation}
    G_a(x) = \sum_{\mu=0}^{d} {(L_a)}_\mu(x) + {(R_a)}_\mu(x-\mu) \, \quad \forall a=1,\ldots,N_c^2-1 \, ,
  \end{equation}
  and the invariance condition to be imposed on physical states reads:
  \begin{equation}
    \label{eq:GaussLawCondition}
    G_a(x) \ket{\psi}_{\text{phys.}} = 0 \,\, .
  \end{equation}

  \begin{proposition}
  The $L_a$ and $R_a$ are related by:
  \begin{align}
    \label{eq:RandLparallelTransportComponents}
    L_a =  {U}^{(1)}_{ba} R_b  \, , \\
    \label{eq:RandLparallelTransportComponents2}
    R_a = {U}^{(1)}_{ab} L_b \, ,
  \end{align}
  where $U^{(1)}_{ab} = -2 \operatorname{Tr}(\tau_a U^\dagger \tau_b U)$.
  \end{proposition}
  \begin{proof}
  From the action on a generic $\ket{\mathcal{U}_\mu(x)}$:
  \begin{equation}
    \label{eq:DRDLrelation}
    D_{R_\mu(x)}[\Phi] \vcentcolon= 
    D_{R_\mu(x)}[\mathcal{U}^\dagger_\mu(x) \, \Omega^\dagger \, \mathcal{U}_\mu(x)] = 
    D_{L_\mu(x)}[\Omega]
    \, .
  \end{equation}
  We now consider an infinitesimal transformation ${\Omega= e^{i \omega_a \tau_a}}$.
  By construction, also ${\Phi = e^{i \phi_a \tau_a}}$ is infinitesimal. 
  After some algebra we find
  ${
    \phi_a = 
    - 2 i \operatorname{Tr}\left(\tau_a \Phi \right)=
    \mathcal{U}^{(1)}_{ab} \omega_b
  }$.
  Using Eq.~\eqref{eq:DRDLrelation} we get:
  \begin{equation}
    e^{i \phi_a R_a} = e^{i \mathcal{U}^{(1)}_{ab} \omega_b R_a} = e^{i \omega_b L_b}
    \implies
    L_a = \mathcal{U}^{(1)}_{ba} R_b \, .
  \end{equation}
  Eq.~\eqref{eq:RandLparallelTransportComponents} follows from the validity on arbitrary linear combinations of $\ket{\mathcal{U}_\mu(x)}$.
  With analogous steps we also find ${R_a = {U}^{(1)}_{ab} L_b}$.
  \end{proof}
  \begin{remark}
  As a consequence, $U^{(1)}_{ab}$ is an orthogonal matrix and therefore:
  \begin{equation} \label{eq:LsquaredEqualsRsquared}
  \sum_a L_a L_a = \sum_a R_a R_a \, .
  \end{equation}
  Physically, the existence of the $R_a$ is important, as it introduces additional quantum numbers
  twin to the $L_a$'s, with the constraint of Eq.~\eqref{eq:LsquaredEqualsRsquared}.
  \end{remark}

  The Hamiltonian of the system is found using the transfer matrix formalism. 
  In the Weyl gauge ($A_0=0$)~\cite{PhysRevD.15.1128}, one finds the
  standard Wilson Hamiltonian~\cite{PhysRevD.15.1128,PhysRevD.20.2610,PhysRevD.15.1111}:
  \begin{equation}
  \label{eq:OriginalKGHamiltonian}
    H = H_{\text{el}} + H_{\text{mag}}
    = 
    \frac{g^{2}}{2}
    \sum_{\vec{x}} \sum_{\mu=1}^{d-1} \sum_{a=1}^{N_c^2 -1} 
    {(L_a)}_{\mu}^{2}(\vec{x})
    -
    \frac{1}{g^{2}}
    \sum_{\vec{x}} \sum_{\mu=1, \nu<\mu}^{d-1} \,
    \mathrm{Tr}[{U}_{\mu \nu}(\vec{x}) + {U}_{\mu \nu}^\dagger(\vec{x}) ] \, .
  \end{equation}
  The products of fields in Eq.~\eqref{eq:OriginalKGHamiltonian} at
  different lattice points have to be understood as tensor products on $\mathcal{H}$.
  In analogy to QED, $H_{\text{el}}$ and $H_{\text{mag}}$ are called respectively the
  \textit{electric} and \textit{magnetic} part of the Hamiltonian.
  Note that Eqs.~\eqref{eq:LaCanonicalCommRel},~\eqref{eq:RaCanonicalCommRel}, 
  which follow from symmetry only, 
  are consistent with Eq.~\eqref{eq:UmuFromAmu}
  (with the $+$ convention) and the usual equal time canonical commutation
  relations $\allowbreak {[{L_{\mu}}_a(\vec{x}), A^b_\nu(\vec{y})] = -i \delta_{ab} \delta_{\mu \nu} \delta(\vec{x}-\vec{y})}$
  (cf. Refs.~\cite{PhysRevD.19.531,10.1143/PTP.55.1631}). 
  For the above reasons, the $L_a$ and $R_a$ are called \textit{canonical momenta}, and loosely speaking
  we can say that ``the links live in the gauge group, while the momenta
  in the algebra''.
  The geometrical picture of Eq.~\eqref{eq:OriginalKGHamiltonian} is that 
  it is the Hamiltonian of particles moving on $\mathrm{SU}(N_c)$ manifolds (one for each lattice link),
  interacting through the magnetic part of the Hamiltonian (see e.g. Ref.~\cite{Polychronakos_2006,POLYCHRONAKOS2023116314}).
  
  The Weyl gauge trivializes the temporal links ($\mathcal{U}_0(x)=\mathds{1}$),
  leaving a residual gauge redundancy associated with the subset of purely spatial transformations of Eq.~\eqref{eq:LocalGaussTransformations}.
  The electric Hamiltonian is the $su(2)$ quadratic Casimir and hence commutes with the $G_a(\vec{x})$. 
  The magnetic part commutes with the latter thanks to Eqs.~\eqref{eq:LaCanonicalCommRel},~\eqref{eq:RaCanonicalCommRel}. 
  For this reason we can, equivalently, define the theory by simply requiring the following:
  \begin{itemize}
  \item
  The (magnetic) basis of the Hilbert space $\mathcal{H}$ is given by the analogous of Eq.~\eqref{eq:UBasisDef}, limited to spatial links:
  \begin{equation}
    \label{eq:UBasisSpatialDef}
  \ket{\mathcal{U}} =  \bigotimes_{\vec{x}} \bigotimes_{\mu=1}^{d-1} \ket{\mathcal{U}_\mu(\vec{x})} \, ,
  \end{equation}
  \item 
  The $U_\mu(\vec{x})$ are commuting $\mathrm{SU}(N_c)$-valued operators:
  \begin{equation}
    {U_\mu(\vec{x}) \ket{\mathcal{U}} \vcentcolon= \mathcal{U}_\mu(\vec{x}) \ket{\mathcal{U}}}.
  \end{equation}
  \item 
  The $L_a$ and $R_a$ form an $su(N_c)$ algebra, they satisfy Eqs.~\eqref{eq:LandRCommute},~\eqref{eq:LaCanonicalCommRel},~\eqref{eq:RaCanonicalCommRel},
  and are related by:
  \begin{equation}
    \label{eq:RandLparallelTransport}
    R = \tau_a R_a = 
    - U^{\dagger} \tau_a L_a U=
    - U^{\dagger} L U
    \, .
  \end{equation}
  Put another way, $R=\tau_a R_a$ is the quantum analog of the parallel transport of $L = \tau_a L_a$.
  We note that, using the normalization of the $\tau_a$ above, Eq.~\eqref{eq:RandLparallelTransport} gives back Eq.~\eqref{eq:RandLparallelTransportComponents2}.
  \item 
  The space of physical states is the subspace of $\mathcal{H}$ satisfying:
  \begin{equation}
    \label{eq:GaussLawSpatialCondition}
    G_a(\vec{x}) \ket{\psi}_{\text{phys.}} = 
    \left[ \sum_{\mu=1}^{d} {(L_a)}_\mu(\vec{x}) + {(R_a)}_\mu(\vec{x}-\mu) \right] \ket{\psi}_{\text{phys.}} 
    = 0 
    \,\, .
  \end{equation}
  In analogy with QED, this is called \textit{Gauss' law}.
  \end{itemize}
  %

\section{Geometrical structure of $\mathrm{SU}(N_c)$}
\label{sec-GeometrysuNc}

In this section we take a look at the canonical momenta from a
gemetrical perspective. 
We work in the wavefunction formalism in the magnetic basis of states $\ket{\mathcal{U}}$:
$U \ket{\mathcal{U}} = \mathcal{U} \ket{\mathcal{U}}$.
Without loss of generality, we discuss only the
construction of the $L_a$. For the $R_a$ the same logic applies.

\begin{lemma}
  \label{lem-XPCommutatorDiffOperator}

Let's consider two operators $X, P$, with the latter acting as a
differential operator $\mathcal{P}$ on the space of wavefunctions of
$X$. 
Then, $[P,X]$ is equivalently given by $\mathcal{P}(x)$ with the replacement $x \to X$.

\end{lemma}

\begin{proof}

We can prove the commutation relation for $[P, X]$ by writing the
commutator action on a generic element of the eigenfunction basis
$\{\psi(x) = \langle x| \psi\rangle \}$. By hypothesis, in the
wavefunction formalism $P$ is a differential operator on the basis of
$X$, and the commutator is equal to: \begin{equation}
  \label{eq:PXcommutatorOnWavefunctions}
  \begin{split}
    [P,X]  \psi(x) &=
    \mathcal{P}(x\psi(x)) - x \cdot \mathcal{P}(\psi(x)) \\
    &= [\mathcal{P}(x)] \psi(x) + x \cdot \mathcal{P}(\psi(x)) - x \cdot \mathcal{P}(\psi(x)) \\
    &= [\mathcal{P}(x)] \psi(x) = P(X) \psi(x) \, .
  \end{split}
  \end{equation}

\end{proof}

\begin{proposition}
\label{prp-LieDerivativesFaithfulRep}

In the manifold wavefunction formalism, a faithful representation of the
momenta is given by the Lie derivatives $\mathcal{L}_a$~\cite{wald2010general} along the directions
of the Lie algebra generators.
Namely, the $\mathcal{L}_a$ form an $su(N_c)$ Lie algebra and  $\mathcal{L}_a \mathcal{U} = - \tau_a \mathcal{U}$.

\end{proposition}

\begin{proof}

The action of a vector field $\dummy{L}_a$ on a differentiable
function $f: \mathrm{SU}(N_c) \to \mathbb{C}$ coincides with the Lie derivative
(see, e.g., appendix C of Ref.~\cite{wald2010general}):
\begin{equation} \label{eq:momenta.Lie.derivatives}
\mathcal{L}_{a} f(p) = \dummy{L}_a(f) |_p = -i \frac{\mathrm{d}}{\mathrm{d} \omega} f(\phi^{\ell_a}_\omega(p)) |_{\omega=0} \, ,
\end{equation}
where $\phi^{\ell_a/r_a}_\omega$ defines the Lie group
flow along the direction of the abstract generators $\ell_a$ in
$su(N_c)$, starting on the point $p$. 

From theorem 13.6 of Ref.~\cite{lee2012smooth} we have
${\mathcal{L}_a \dummy{L}_b = [\dummy{L}_a, \dummy{L}_b}]$, therefore the $\mathcal{L}_a$ form an $su(N_c)$ Lie algebra (cf.~Eqs.~\eqref{eq:LaLieAlgebraCommRel},~\eqref{eq:RaLieAlgebraCommRel}).
Finally, in the particular case $f(p)=\mathcal{U}(p)$ we get:
\begin{equation}
\mathcal{L}_{a} \mathcal{U} = -i \frac{\mathrm{d}}{\mathrm{d} \omega} (e^{-i \omega \tau_a} \mathcal{U})|_{\omega=0} = - \tau_a \mathcal{U}
\end{equation}

\end{proof}

\begin{remark}
Using lemma~\eqref{lem-XPCommutatorDiffOperator}, the we see that the $\mathcal{L}_{a}$ lead to the right canonical commutation relations 
of. Eqs.~\eqref{eq:LaCanonicalCommRel},~\eqref{eq:RaCanonicalCommRel}.
\end{remark}

\begin{proposition}
  \label{prp-LieDerivativesComponentsMaurerCartan}
  Let $L_a$ be the vector fields on the  manifold $\mathrm{SU}(N_c)$ defined by:
  \begin{equation}
    \dummy{L}_a \mathcal{U} = -\tau_a \mathcal{U} \, . 
  \end{equation}
  $\dummy{L}_a$ acts componentwise on $\mathcal{U}$.
  Consider now a chart of $\mathrm{SU}(N_c)$, with a set of coordinates $x^k \, (k=1,\ldots,N_c^2-1)$.
  The vector fields $\dummy{L}_a$ read:
    \begin{equation}
    \dummy{L}_a = \dummy{L}_a^k \frac{\partial}{\partial x^k} \, .
    \end{equation}
  Let also $m$ be the Maurer-Cartan 1-form (see e.g.~appendix A of Ref.~\cite{Shnir:2005xx}):
  \begin{equation}
    m = \mathcal{U}^{-1} d\mathcal{U} 
    = \mathcal{U}^{-1} \frac{\partial \mathcal{U}}{\partial x^k} \mathrm{d}x^k 
    = m_k \mathrm{d}x^k \, ,
  \end{equation}
  where $m_k = i c^a_k \tau_a \in su(N_c)$ (see e.g.~Sec.~A.1.4 of Ref.~\cite{gattringer2009quantum}).
  Then, the components $\dummy{L}_a^k$ are the solutions of the
  following linear system: 
  \begin{equation}
    \dummy{L}_a^k c^b_k = 2 i \operatorname{Tr} \left[\mathcal{U}^{-1} \tau_a \mathcal{U} \tau_b \right] \, , \, \forall a,b = 1, \ldots, N_c^2-1 \, .
  \end{equation}
  \end{proposition}
  
  \begin{proof}

    The proof follows immediately by applying the vector fiels to the Maurer-Cartan form. In fact, from the normalization
    $\frac{\partial}{\partial x^i} (\mathrm{d}x^k) = \delta_i^k$, we get:
    \begin{align}
      \dummy{L}_a(m) &= i \dummy{L}_a^k c^b_k \tau_b  \, ,
      \\
      \dummy{L}_a(m) &=  \mathcal{U}^{-1} \dummy{L}_a^k \frac{\partial \mathcal{U}}{\partial x^k} = \mathcal{U}^{-1} \dummy{L}_a \mathcal{U} = - \mathcal{U}^{-1} \tau_a \mathcal{U} \, .
    \end{align} 
    Thus, by using $\operatorname{Tr}(\{\tau_a, \tau_b\}) = \delta_{ab}$,
    we get the aforementioned linear system with the $\dummy{L}_a^k$ as solutions: 
    \begin{equation}
      \dummy{L}_a^k c^b_k = 2 i \operatorname{Tr} \left[\mathcal{U}^{-1} \tau_a \mathcal{U} \tau_b \right] \, , \, \forall a,b = 1, \ldots, N_c^2-1 \, .
    \end{equation}
  \end{proof}
  
\begin{remark}

For $\mathrm{SU}(2)$, the $\dummy{L}_a$ are the Killing vectors of the
Atiyah-Hitchin metric~\cite{GIBBONS1986183}.
Their form is given by Eqs.~\eqref{eq:L1.diff.op.continuum},
\eqref{eq:L2.diff.op.continuum},~\eqref{eq:L3.diff.op.continuum}, and
$\sum_a \dummy{L}_a \dummy{L}_a$ is the Laplace-Beltrami operator on
$S_3$~\cite{chin1985exact}.

\end{remark}

We now look at the particular case of $\mathrm{SU}(2)$.
\begin{proposition}
\label{prp-CommRelFromDiffOps}
  
  The differential operators of Eqs.~\eqref{eq:L1.diff.op.continuum},~\eqref{eq:L2.diff.op.continuum},~\eqref{eq:L3.diff.op.continuum} 
  give rise to the canonical commutation relations of Eq.~\eqref{eq:CanonicalCommRel}.
  
  \end{proposition}
  
  \begin{proof}

  The proof follows immediately from lemma~\eqref{lem-XPCommutatorDiffOperator} by using the expression of $U$ 
  (see. Eq.~\eqref{eq:su2UalphaMatrix}) and the properties of the momenta (see \Crefrange{eq:su2IrrepsOperators}{eq:su2Rpmconvention}). 
  In the following steps $\psi(U)$ is a generic wavefunctional of the group manifold elements.
  
  The commutator $[L_3, U]$ acts as: 
  \begin{equation}
    \label{eq:L3UCommutatorDicreteCase}
    \begin{split}
    [L_3, U] \psi(U) &=
    [L_3 \, \mathcal{U}(\vec{\alpha})] \psi(U) =
     \begin{pmatrix}
      -\frac{1}{2} D^{1/2}_{-1/2, -1/2}(\vec{\alpha}) &
      +\frac{1}{2} D^{1/2}_{-1/2, +1/2}(\vec{\alpha}) \\
      -\frac{1}{2} D^{1/2}_{+1/2, -1/2}(\vec{\alpha}) &
      +\frac{1}{2} D^{1/2}_{+1/2, +1/2}(\vec{\alpha})
     \end{pmatrix}
     \psi(U) \\
     &= -\tau_3 \, \mathcal{U}(\vec{\alpha}) \psi(U) =
     - \tau_3 U \psi(U)
     \, .
     \end{split}
    \end{equation}
  
  For $L_1$ and $L_2$ we can equivalently prove the commutation
  relations by looking at ${L_\pm = L_1 \pm i L_2}$: \begin{equation}
    \label{eq:LplusUCommutatorDiscreteCase}
    \begin{split}
    [L_+, U] \psi(U) 
    &= 
    [L_+ \, \mathcal{U}(\vec{\alpha})] \psi(U)
     =
     \begin{pmatrix}
      D^{1/2}_{1/2, 1/2}(\vec{\alpha}) &
      -D^{1/2}_{1/2, -1/2}(\vec{\alpha}) \\
      0 & 0
     \end{pmatrix}
     \psi(U) \\ 
     &= -\tau_{+} \, \mathcal{U}(\vec{\alpha}) \psi(U) =
     - \tau_{+} U \psi(U)
     \, ,
     \end{split}
    \end{equation}
    and \begin{align}
    \label{eq:LminusUCommutatorDiscreteCase}
    \begin{split}
    [L_-, U] \psi(U) 
    &= 
    [L_- \, \mathcal{U}(\vec{\alpha})] \psi(U)
     =
     \begin{pmatrix}
      0 & 0 \\
      -D^{1/2}_{-1/2, 1/2}(\vec{\alpha}) &
      D^{1/2}_{-1/2, -1/2}(\vec{\alpha})
     \end{pmatrix}
     \psi(U) \\ 
     &= 
     - \tau_{-} \, \mathcal{U}(\vec{\alpha}) \psi(U) =
     -\tau_{-} U \psi(U)
     \, ,
     \end{split}
    \end{align} where $\tau_\pm = \tau_1 \pm i \tau_2$. It follows that
  $[L_{1,2}, U] = -\tau_{1,2} U$.
  
  The proof for the $R_a$ are analogous. We remark that in the latter
  case one must use the convention of \Cref{eq:su2R3convention,eq:su2Rpmconvention} in order to get Eq.~\eqref{eq:CanonicalCommRel}.
  
  \end{proof}

\section{Finite differences on $S_3$}
\label{sec-FiniteDifferences}

In this section we provide an alternative construction of the
finite-dimensional canonical momenta. This is done constructing finite
difference operators such that the maximum number of continuum
eigenstates is reproduced. 

In the continuum manifold, the canonical momenta are represented by the
Lie derivatives along the generators of the Lie algebra (see App.~\eqref{sec-GeometrysuNc} for a derivation). 
Finding a finite dimensional
representation of the $\dummy{L}_a$ means approximating the latter by
finite differences operators. The convergence to continuous differential
operators is a long standing problem in mathematics (see e.g. Ref.~\cite{9d7cc18a-d051-3230-be5c-06598cab3b15}),
and in general depends on the set of functions the operators act on. In
our case we focus on the dicrete versions of the continuum
eigenfunctions, as our physical state will be a linear combination of
the latter.

The finite differences are not unique. One can use Taylor
approximations, restrict the space of functions, etc. In Ref.~\cite{Jakobs:2023lpp} for instance, we
constructed the momenta using a Delaunay triangulation of $S_3$. An
interesting choice is to consider operators preserving the spectrum on a
subspace of functions. Some examples are Ref.~\cite{Calogero:1982xa} for polynomials and
Ref.~\cite{Torres-Vega17} for exponential functions. 
Here we are interested in applying the results of Ref.~\cite{Calogero:1984ew,Calogero1985}
for trigonometric polynomials, generalizing the application to orbital
angular momentum of Ref.~\cite{campos2000finite}.

Let $f(x)$ be a $\tau$-periodic trigonometric polynomial sampled at
$N$ points ${x_1,\ldots,x_N}$ in the interval $(0, \tau)$. The
vector $\vec{f}$, with components
${f_i = f(x_i)}$ (${i=1,\ldots,N}$), is the discrete version of the
function. If the degree of $f(x)$ is at most ${n = \frac{N-1}{2}}$,
we can write a finite-difference operator $D_x$ that gives exact
derivatives on $\vec{f}$ (see e.g.~Sec.~2.3 of Ref.~\cite{campos2019xft}):
\begin{equation}
  \label{eq:DxFiniteDiff}
  (D_x)_{ij} = \begin{cases}
    \frac{1}{2} \sum_{\ell \neq i}^{N} \cot\frac{(x_i - x_\ell)}{2}, & \text{if } i = j, \\
    \frac{1}{2} \frac{s'(x_i)}{s'(x_j)} \csc\frac{(x_i - x_j)}{2}, & \text{if } i \neq j \, ,
  \end{cases}
\end{equation}
where
${ s(x) = \prod_{\ell=1}^{N} \sin{\frac{x - x_\ell}{2}} }$. For
anti-periodic functions the derivatives are still exact by setting
$f_1 = 0$ and hold for degrees up to ${n = \frac{N-2}{2}}$.
Physically this is a valid dicretization if the points $x_i$ become
asymptotically dense with $N$, because then
$\lim_{N \to \infty} x_1 = 0$ and $f(0) = 0$ due to
anti-periodicity. Functions $g(x)$ are represented by diagonal
matrices with elements ${G_{ij} = g(x_i) \delta_{ij}}$. We also remark
that higher derivatives are obtained by higher powers of $D_x$, and
multi-variable differential operators are found as tensor products of
the respective operators~\cite{campos2000finite}.

\begin{proposition} \label{prp-LaTrigPolyFiniteDiff}

Let's consider the finite-dimensional operators obtained from Eqs.~\eqref{eq:L1.diff.op.continuum},~\eqref{eq:L2.diff.op.continuum},
\eqref{eq:L3.diff.op.continuum} and~\eqref{eq:Lsquared.diff.op.continuum} 
by the replacements
$\frac{\partial}{\partial \alpha_i} \to D_{\alpha_i}$, where
$\vec{\alpha} = (\theta, \phi, \psi)$ and $D_{\alpha_i}$ being the
analog of Eq.~\eqref{eq:DxFiniteDiff}. 
On $S_3$, if ${N_\theta \geq 2q+1}$
and $N_\phi$, $N_\psi$ satisfy the constraints of theorem
(\ref{thm-DJTUnitarity}), these are exact derivatives of the Wigner
D-functions $D^j_{m_L, m_R}$ with $j \leq q$.

The same applies to the $R_a$.

\end{proposition}

\begin{proof}

Analogously to the case of orbital angular momentum~\cite{campos2000finite}, the form of the
continuum eigenfunctions determines the number of eigenstates we
reproduce exactly. This is found by determining the minimum number of points
given the maximum degree of the trigonometric polynomial, and
translating it to the original interval.

The $D^j_{m_L m_R}$ are trigonometric polynomials in $\theta/2$ of maximum degree $2q$ (see Eq.~\eqref{eq:WignerdFunctionsDef}). 
Therefore, if $j \leq q$, the
derivatives are exact by considering at least $2q + 1$
points.
Equivalently, the same number points can be taken directly in $(0,\pi)$, 
as $D^j_{m_L m_R}$ is fully determined by the values in the interval $(0, \pi/2)$.

We now come to the differentiation with respect to $\phi$ and
$\psi$. The eigenfunctions $D^j_{m_L, m_R}$ (see Eq.~\eqref{eq:WignerDFunctionsDefinition}) have now become tensor product
vectors. Therefore we need to ensure $N_\phi, N_\psi \geq N_\theta$ in
order to keep the $\theta$-derivatives exact. 
We see that this condition is fulfilled, as the exponentials
$e^{i m_L \phi}$, $e^{i m_R \psi}$ are $2\pi$-periodic in $\phi/2$
and $\psi/2$, with maximum degree $2j$. 
Therefore, the minimum number of points is $4q + 1$ for
$\phi/2, \psi/2 \in (0, 2\pi)$.

\end{proof}

We remark that the number of points $N_\theta$ for the $\theta$ direction is higher than the DJT, 
though sharing the same feature of the exact behavior on a subspace.
By the same arguments used in Sec.~\eqref{sec-su2MatrixOperators} it
follows that prop.~\eqref{prp-FirstNqEigenstates} and
(\ref{prp-FirstNqprimeCommRel}) hold also for the $L_a$ constructed as
above.

The generalization of this method to $\mathrm{SU}(3)$ is obtained by replacing the partial
derivatives of the continuum manifold \cite{byrd1997geometry} with the finite difference
operator of Eq.~\eqref{eq:DxFiniteDiff}.


\section{Properties for the DJT construction of the momenta}

We report here some results used in the construction of our canonical momenta matrix representation with the DJT 
(cf.~Sec.~\eqref{sec-su2MatrixOperators}).

\begin{lemma}
\label{lem-wsDistribution}

Let $P_n(x)$ be the $n$-th Legendre Polynomial, with roots
$x_s = \cos{(\theta_s)}$. As $n \to \infty$, the $\theta_s$ become
evenly spaced in the interval $[0, \pi]$, and the Gaussian weights
$w_s \sim \frac{\pi}{n} \sin(\theta_s)$.

\end{lemma}

\begin{proof}
In order to prove this, we show that as $n \to \infty$ the roots
$x_s$ of $P_n(x)$ converge to the evenly spaced points used in the
rectangle rule of integration with $n$ points:
\begin{equation}
  \sum_i \Delta z_i\, g(z) \to \int_a^b \mathrm{d}z\, g(z) \, , \qquad \Delta z = (b-a)/n \, .
\end{equation}
$g$ is a generic integrable function in the interval
$[a,b]$. The asymptotic behavior of the $x_s$ is~\cite{LETHER1995245}: 
\begin{equation}
  x_s = 
  \left( 1 - \frac{1}{8 n^2} +  \frac{1}{8 n^3} \right) 
  \cdot 
  \cos{\left(\pi \frac{s - 1/4}{n + 1/2}\right)} 
  + O\left( \frac{1}{n^4} \right) 
  \, ,
\end{equation}
namely:
\begin{equation}
  x_s = \cos{\theta_s} =
  \cos{\left(\pi \frac{s}{n}\right)}
  \left[ 1 + O\left(\frac{1}{n}\right) \right] 
  \, .
\end{equation}
Therefore as $n \to \infty$,
$\theta_s \sim \pi \frac{s}{n}$. Using this result, in the same limit
we get:
\begin{equation}
\sum_s w_s\, f(x_s) \sim 
\int_{-1}^{1} \mathrm{d}x\, f(x) =
\int_{0}^{\pi} \mathrm{d}\theta\, \sin{(\theta)}\, f(\cos{(\theta)}) \sim
\sum_s \Delta \theta\, \sin(\theta_s)\, f(\cos{(\theta_s)})
\, ,
\end{equation}
which gives
$w_s \sim \Delta \theta \sin{\theta_s} \sim \Delta x$.

\end{proof}

\begin{lemma}
  \label{lem-FirstNqprimeAsContinuum}
  
  Let $\vec{w}^j_{m_L m_R}$ be the vectors made by the values of the
  Wigner D-functions $D^{j}_{m_L m_R}$ at the points of the partitioning
  of the sphere $S_3$. Then, the action of the matrices in Eqs.~\eqref{eq:LaSimilarityTransform},\eqref{eq:RaSimilarityTransform} on the
  first $N_{q'}=N_{q-1/2}$ eigenstates of ${\sum_a L_a L_a}$ is the
  same as their differential operators in the continuum manifold.
  
  \end{lemma}
  
  \begin{proof}
  The components of $\mathcal{U}(\vec{\alpha})$ in Eq.~\eqref{eq:su2UalphaOperatorMatrix} are just a particular case of the
  Wigner $D$-functions, with $j=1/2$. The $N_\alpha$ values are on
  the diagonal (in the Hilbert space).
  
  Let's now consider a $N_\alpha \times N_\alpha$ matrix W with
  components\footnote{We observe that ${\operatorname{DJT} \cdot W}$
  resembles what Ref.~\cite{kostelec2008ffts} calls ``Discrete 
  SO(3) Fourier Transform''.}:
  \begin{equation}
    W_{i,j} = \sqrt{w_{s(i)}} \delta_{ij} \, ,
  \end{equation}
  where $s(i)$ is the $\theta$ index correspoding to
  $\vec{\alpha}_i$. $W$ is invertible because $w_s > 0$ from theorem~\eqref{thm-PolynomialWeights}. 
  We observe that if $\vec{v}^j_{m_L m_R}$
  is the $(j,m_L,m_R)$ column of the $\operatorname{DJT}$ in Eq.~\eqref{eq:DJTMatrix}, the vectors
  $\vec{w}^j_{m_L m_R} = W^{-1} \vec{v}^j_{m_L m_R}$ form a basis for the
  first $N_{q - 1/2}$ $su(2)$ irreps. In fact the
  $\operatorname{DJT}$ of Eq.~\eqref{eq:DJTMatrix} has orthonormal
  columns (see theorem~\eqref{thm-DJTUnitarity}), and the
  $\vec{w}^j_{m_L m_R}$ are obtained by a change of basis transformation
  $W^{-1}$ from the orthonormal basis.
  
  If we consider $U \, \vec{w}^j_{m_L m_R}$, this will be a quadruplet of
  vectors (one for each component of
  $\mathcal{U}(\vec{\alpha}) \vec{w}^j_{m_L m_R}$) whose components are
  proportional to the $N_\alpha$ values of
  $D^{1/2}_{\mp 1/2, \pm 1/2}(\vec{\alpha}) \cdot D^{j}_{m_L m_R}(\vec{\alpha})$.
  The product of two Wigner $D$-functions of degree $j_1$ and $j_2$
  is a linear combination of $D$-functions of degree
  ${J \leq j_1 + j_2}$ (see e.g.~subsec. 4.6.1 of Ref.~\cite{doi:10.1142/0270}). 
  Therefore the action of the matrices $L_a$ and $R_a$ on the first $N_{q-1/2}$
  vectors $\vec{w}^j_{m_L m_R}$ is the same as the continuous manifold.
  Since these vectors form a basis, this is a property of the whole
  subspace of the first $N_{q-1/2}$ eigenstates.
  
  \end{proof}

\end{appendices}

\addcontentsline{toc}{section}{References}

\printbibliography{}

@ARTICLE{1697831,
  author={Shannon, C.E.},
  journal={Proceedings of the IRE}, 
  title={Communication in the Presence of Noise}, 
  year={1949},
  volume={37},
  number={1},
  pages={10-21},
  doi={10.1109/JRPROC.1949.232969}
}

@book{sansone1959orthogonal,
  title={Orthogonal functions},
  author={Sansone, Giovanni},
  volume={9},
  year={1959},
  publisher={Interscience Publishers}
}

@article{PhysRevD.11.395,
  title = {Hamiltonian formulation of {Wilson's} lattice gauge theories},
  author = {Kogut, John and Susskind, Leonard},
  journal = {Phys. Rev. D},
  volume = {11},
  issue = {2},
  pages = {395--408},
  numpages = {0},
  year = {1975},
  month = {Jan},
  publisher = {American Physical Society},
  doi = {10.1103/PhysRevD.11.395},
  url = {https://link.aps.org/doi/10.1103/PhysRevD.11.395}
}

@article{santhanam1976quantum,
  title={Quantum mechanics in finite dimensions},
  author={Santhanam, TS and Tekumalla, AR},
  journal={Foundations of Physics},
  volume={6},
  number={5},
  pages={583--587},
  year={1976},
  publisher={Springer},
  doi = {10.1007/BF00715110}
}

@article{10.1143/PTP.55.1631,
    author = {Utiyama, Ryoyu and Sakamoto, Jiro},
    title = "{Canonical Quantization of Non-Abelian Gauge Fields}",
    journal = {Progress of Theoretical Physics},
    volume = {55},
    number = {5},
    pages = {1631-1648},
    year = {1976},
    month = {05},
    issn = {0033-068X},
    doi = {10.1143/PTP.55.1631},
    url = {https://doi.org/10.1143/PTP.55.1631},
    eprint = {https://academic.oup.com/ptp/article-pdf/55/5/1631/5427808/55-5-1631.pdf},
}

@article{PhysRevD.15.1128,
  title = {Gauge fixing, the transfer matrix, and confinement on a lattice},
  author = {Creutz, Michael},
  journal = {Phys. Rev. D},
  volume = {15},
  issue = {4},
  pages = {1128--1136},
  numpages = {0},
  year = {1977},
  month = {Feb},
  publisher = {American Physical Society},
  doi = {10.1103/PhysRevD.15.1128},
  url = {https://link.aps.org/doi/10.1103/PhysRevD.15.1128}
}

@article{PhysRevD.15.1111,
  title = {Strong-coupling calculations of the hadron spectrum of quantum chromodynamics},
  author = {Banks, T. and Raby, S. and Susskind, L. and Kogut, J. and Jones, D. R. T. and Scharbach, P. N. and Sinclair, D. K.},
  journal = {Phys. Rev. D},
  volume = {15},
  issue = {4},
  pages = {1111--1127},
  numpages = {0},
  year = {1977},
  month = {Feb},
  publisher = {American Physical Society},
  doi = {10.1103/PhysRevD.15.1111},
  url = {https://link.aps.org/doi/10.1103/PhysRevD.15.1111}
}

@article{PhysRevD.19.531,
  title = {Gauge fixing and canonical quantization},
  author = {Creutz, Michael and Muzinich, I. J. and Tudron, Thomas N.},
  journal = {Phys. Rev. D},
  volume = {19},
  issue = {2},
  pages = {531--539},
  numpages = {0},
  year = {1979},
  month = {Jan},
  publisher = {American Physical Society},
  doi = {10.1103/PhysRevD.19.531},
  url = {https://link.aps.org/doi/10.1103/PhysRevD.19.531}
}

@article{PhysRevD.20.2610,
  title = {Lattice models of quark confinement at high temperature},
  author = {Susskind, Leonard},
  journal = {Phys. Rev. D},
  volume = {20},
  issue = {10},
  pages = {2610--2618},
  numpages = {0},
  year = {1979},
  month = {Nov},
  publisher = {American Physical Society},
  doi = {10.1103/PhysRevD.20.2610},
  url = {https://link.aps.org/doi/10.1103/PhysRevD.20.2610}
}

@article{Nielsen:1981hk,
    author = "Nielsen, Holger Bech and Ninomiya, M.",
    title = "{No Go Theorem for Regularizing Chiral Fermions}",
    reportNumber = "RL-81-052",
    doi = "10.1016/0370-2693(81)91026-1",
    journal = "Phys. Lett. B",
    volume = "105",
    pages = "219--223",
    year = "1981"
}

@article{Calogero:1982xa,
    author = "Calogero, F.",
    title = "{Lagrangian Interpolation and Differentiation}",
    reportNumber = "ROME-306-1982",
    doi = "10.1007/BF02754737",
    journal = "Lett. Nuovo Cim.",
    volume = "35",
    pages = "273",
    year = "1982",
    note = "[Erratum: Lett.Nuovo Cim. 36, 447 (1983)]"
}

@article{king1982introduction,
  title={Introduction to numerical analysis},
  author={Stoer, J and Bulirsch, R},
  journal={SIAM Review},
  volume={24},
  number={1},
  pages={96},
  year={1982},
  publisher={Society for Industrial and Applied Mathematics},
  doi = {10.1007/978-0-387-21738-3}
}

@article{Ginsparg:1981bj,
    author = "Ginsparg, Paul H. and Wilson, Kenneth G.",
    title = "{A Remnant of Chiral Symmetry on the Lattice}",
    reportNumber = "CLNS-81-520, HUTP-81-A060",
    doi = "10.1103/PhysRevD.25.2649",
    journal = "Phys. Rev. D",
    volume = "25",
    pages = "2649",
    year = "1982"
}

@article{STOVICEK1984157,
  title = {Quantum mechanics in a discrete space-time},
  journal = {Reports on Mathematical Physics},
  volume = {20},
  number = {2},
  pages = {157-170},
  year = {1984},
  issn = {0034-4877},
  doi = {10.1016/0034-4877(84)90030-2},
  url = {https://www.sciencedirect.com/science/article/pii/0034487784900302},
  author = {P. Šťovíček and J. Tolar}
}

@article{Calogero:1984ew,
    author = "Calogero, F.",
    title = "{Interpolation, differentiation and solution of eigenvalue problems for periodic functions}",
    reportNumber = "ROME-384-1984",
    doi = "10.1007/BF02813629",
    journal = "Lett. Nuovo Cim.",
    volume = "39",
    pages = "305",
    year = "1984"
}

@article{Calogero1985,
  author={Calogero, F.},
  title={Interpolation and differentiation for periodic functions},
  journal={Lettere al Nuovo Cimento (1971-1985)},
  year={1985},
  month={Feb},
  day={01},
  volume={42},
  number={3},
  pages={106-110},
  issn={1827-613X},
  doi={10.1007/BF02748342},
  url={https://doi.org/10.1007/BF02748342}
}

@article{9d7cc18a-d051-3230-be5c-06598cab3b15,
 ISSN = {00255718, 10886842},
 URL = {http://www.jstor.org/stable/2007849},
 author = {V. T.},
 journal = {Mathematics of Computation},
 number = {178},
 pages = {834--835},
 publisher = {American Mathematical Society},
 reviewed-author = {G. D. Smith},
 urldate = {2023-08-28},
 volume = {48},
 year = {1987},
 doi = {10.2307/2007849}
}

@article{chin1985exact,
  title={Exact ground-state properties of the SU(2) Hamiltonian lattice gauge theory},
  author={Chin, SA and Van Roosmalen, OS and Umland, EA and Koonin, SE},
  journal={Physical Review D},
  volume={31},
  number={12},
  pages={3201},
  year={1985},
  publisher={APS},
  doi = {10.1103/PhysRevD.31.3201}
}

@article{GIBBONS1986183,
  title = {Classical and quantum dynamics of BPS monopoles},
  journal = {Nuclear Physics B},
  volume = {274},
  number = {1},
  pages = {183-224},
  year = {1986},
  issn = {0550-3213},
  doi = {10.1016/0550-3213(86)90624-3},
  url = {https://www.sciencedirect.com/science/article/pii/0550321386906243},
  author = {G.W. Gibbons and N.S. Manton}
}

@book{doi:10.1142/0270,
  author = {Varshalovich, D A and Moskalev, A N and Khersonskii, V K},
  title = {Quantum Theory of Angular Momentum},
  publisher = {WORLD SCIENTIFIC},
  year = {1988},
  doi = {10.1142/0270},
  address = {},
  edition   = {},
  URL = {https://www.worldscientific.com/doi/abs/10.1142/0270},
  eprint = {https://www.worldscientific.com/doi/pdf/10.1142/0270}
}

@article{BARONE199029,
  title = {A discrete orthogonal transform based on spherical harmonics},
  journal = {Journal of Computational and Applied Mathematics},
  volume = {33},
  number = {1},
  pages = {29-34},
  year = {1990},
  issn = {0377-0427},
  doi = {10.1016/0377-0427(90)90253-V},
  url = {https://www.sciencedirect.com/science/article/pii/037704279090253V},
  author = {Piero Barone}
}

@article{Sommer_1994,
	title = {{A new way to set the energy scale in lattice gauge theories and its application to the static force and $\alpha_s$ in $SU(2)$ Yang-Mills theory}},
	author = {{R. Sommer}},
	journal = {Nucl. Phys. B}, 
	year = 1994,
	month = {Jan},
	publisher = {Elsevier {BV}},
	volume = {411},
	number = {2-3},
	doi = {10.1016/0550-3213(94)90473-1},
	url = {https://doi.org/10.1016\%2F0550-3213\%2894\%2990473-1}
}

@article{LETHER1995245,
  title = {Minimax approximations to the zeros of Pn(x) and Gauss-Legendre quadrature},
  journal = {Journal of Computational and Applied Mathematics},
  volume = {59},
  number = {2},
  pages = {245-252},
  year = {1995},
  issn = {0377-0427},
  doi = {10.1016/0377-0427(94)00030-5},
  url = {https://www.sciencedirect.com/science/article/pii/0377042794000305},
  author = {F.G. Lether and P.R. Wenston}
}

@misc{sakurai1995modern,
  title={Modern quantum mechanics, revised edition},
  author={Sakurai, Jun John and Commins, Eugene D},
  year={1995},
  publisher={American Association of Physics Teachers},
  doi = {10.1017/9781108499996}
}

@book{sakurai_napolitano_2017, 
  place={Cambridge}, 
  edition={2}, 
  title={Modern Quantum Mechanics}, 
  DOI={10.1017/9781108499996}, 
  publisher={Cambridge University Press}, 
  author={Sakurai, J. J. and Napolitano, Jim}, 
  year={2017}
}

@article{byrd1997geometry,
  title={The geometry of $\mathrm{SU}(3)$},
  author={Byrd, Mark},
  journal={arXiv preprint physics/9708015},
  year={1997},
  doi = {10.48550/arXiv.physics/9708015}
}

@article{Chandrasekharan:1996ih,
    author = "Chandrasekharan, S. and Wiese, U. J.",
    title = "{Quantum link models: A Discrete approach to gauge theories}",
    eprint = "hep-lat/9609042",
    archivePrefix = "arXiv",
    reportNumber = "MIT-CTP-2573, CTP-2573",
    doi = "10.1016/S0550-3213(97)00006-0",
    journal = "Nucl. Phys. B",
    volume = "492",
    pages = "455--474",
    year = "1997"
}

@inproceedings{Lepage:1998dt,
    author = "Lepage, G. P.",
    title = "{Lattice QCD for novices}",
    booktitle = "{13th Annual HUGS AT CEBAF (HUGS 98)}",
    eprint = "hep-lat/0506036",
    archivePrefix = "arXiv",
    pages = "49--90",
    month = "5",
    year = "1998",
    doi = {10.48550/arXiv.hep-lat/0506036}
}

@article{Luscher:1998pqa,
    author = "Luscher, Martin",
    title = "{Exact chiral symmetry on the lattice and the Ginsparg-Wilson relation}",
    eprint = "hep-lat/9802011",
    archivePrefix = "arXiv",
    reportNumber = "DESY-98-014",
    doi = "10.1016/S0370-2693(98)00423-7",
    journal = "Phys. Lett. B",
    volume = "428",
    pages = "342--345",
    year = "1998"
}

@article{campos2000finite,
  title={A finite-dimensional representation of the quantum angular momentum operator},
  author={Campos, Rafael G and Pimentel, LO},
  journal={arXiv preprint quant-ph/0008120},
  year={2000},
  doi = {10.48550/arXiv.quant-ph/0008120}
}

@article{de2003quantum,
  title={Quantum mechanics in finite-dimensional Hilbert space},
  author={De la Torre, AC and Goyeneche, D},
  journal={American Journal of Physics},
  volume={71},
  number={1},
  pages={49--54},
  year={2003},
  publisher={American Association of Physics Teachers},
  doi = {10.1119/1.1514208}
}

@book{Shnir:2005xx,
    author = "Shnir, Ya. M.",
    title = "{Magnetic monopoles}",
    year = "2005",
    url = {https://link.springer.com/book/10.1007/3-540-29082-6},
    NOTE = "\url{https://link.springer.com/content/pdf/bbm:978-3-540-29082-7/1}",
    doi = {10.1007/3-540-29082-6}
}

@book{wald2010general,
  title={General relativity},
  author={Wald, Robert M},
  year={2010},
  publisher={University of Chicago press},
  doi = {10.7208/chicago/9780226870373.001.0001}
}

@article{Tsuchiya:2020hur,
    author = "Tsuchiya, Masataka and Houri, Tsuyoshi and Yoo, Chul-Moon",
    title = "{The first-order symmetry operator on gravitational perturbations in the 5D Myers\textendash{}Perry spacetime with equal angular momenta}",
    eprint = "2011.03973",
    archivePrefix = "arXiv",
    primaryClass = "gr-qc",
    doi = "10.1093/ptep/ptab017",
    journal = "PTEP",
    volume = "2021",
    number = "3",
    pages = "033E01",
    year = "2021"
}

@book{lee2012smooth,
  title={Smooth manifolds},
  author={Lee, John M and Lee, John M},
  year={2012},
  publisher={Springer},
  doi = {10.1007/978-1-4419-9982-5}
}

@article{PhysRevD.91.054506,
  title = {Formulation of lattice gauge theories for quantum simulations},
  author = {Zohar, Erez and Burrello, Michele},
  journal = {Phys. Rev. D},
  volume = {91},
  issue = {5},
  pages = {054506},
  numpages = {15},
  year = {2015},
  month = {Mar},
  publisher = {American Physical Society},
  doi = {10.1103/PhysRevD.91.054506},
  url = {https://link.aps.org/doi/10.1103/PhysRevD.91.054506}
}

@article{PhysRevD.101.114502,
  title = {Loop, string, and hadron dynamics in SU(2) Hamiltonian lattice gauge theories},
  author = {Raychowdhury, Indrakshi and Stryker, Jesse R.},
  journal = {Phys. Rev. D},
  volume = {101},
  issue = {11},
  pages = {114502},
  numpages = {25},
  year = {2020},
  month = {Jun},
  publisher = {American Physical Society},
  doi = {10.1103/PhysRevD.101.114502},
  url = {https://link.aps.org/doi/10.1103/PhysRevD.101.114502}
}

@article{murata2008separability,
  title={On the separability of field equations in Myers--Perry spacetimes},
  author={Murata, Keiju and Soda, Jiro},
  journal={Classical and Quantum Gravity},
  volume={25},
  number={3},
  pages={035006},
  year={2008},
  publisher={IOP Publishing},
  doi = {10.1088/0264-9381/25/3/035006}
}

@article{kostelec2008ffts,
  title={FFTs on the rotation group},
  author={Kostelec, Peter J and Rockmore, Daniel N},
  journal={Journal of Fourier analysis and applications},
  volume={14},
  pages={145--179},
  year={2008},
  publisher={Springer},
  doi = {10.1007/s00041-008-9013-5}
}

@article{Polychronakos_2006,
doi = {10.1088/0305-4470/39/41/S07},
url = {https://dx.doi.org/10.1088/0305-4470/39/41/S07},
year = {2006},
month = {sep},
publisher = {},
volume = {39},
number = {41},
pages = {12793},
author = {Alexios P Polychronakos},
title = {The physics and mathematics of Calogero particles},
journal = {Journal of Physics A: Mathematical and General}
}

@article{Alama+2008+137+142,
  author = {Jesse Alama},
  doi = {doi:10.2478/v10037-007-0015-6},
  url = {https://doi.org/10.2478/v10037-007-0015-6},
  title = {The Rank+Nullity Theorem},
  journal = {Formalized Mathematics},
  number = {3},
  volume = {15},
  year = {2008},
  pages = {137--142}
}

@book{gattringer2009quantum,
  title={{Quantum Chromodynamics on the Lattice - an Introductory Presentation}},
  author={Gattringer, Christof and Lang, Christian},
  volume={788},
  year={2009},
  publisher={Springer Science \& Business Media},
  doi = {10.1007/978-3-642-01850-3}
}

@ARTICLE{6006544,
  author={McEwen, Jason D. and Wiaux, Yves},
  journal={IEEE Transactions on Signal Processing}, 
  title={A Novel Sampling Theorem on the Sphere}, 
  year={2011},
  volume={59},
  number={12},
  pages={5876-5887},
  doi={10.1109/TSP.2011.2166394}
}

@book{hamermesh2012group,
  title={Group theory and its application to physical problems},
  author={Hamermesh, Morton},
  year={2012},
  publisher={Courier Corporation},
  doi = {10.1017/S0013091500025773}
}

@article{Celeghini:2014uxa,
    author = "Celeghini, E. and del Olmo, M. A. and Velasco, M. A.",
    title = "{Lie Groups of Jacobi polynomials and Wigner d-matrices}",
    eprint = "1402.5217",
    archivePrefix = "arXiv",
    primaryClass = "math-ph",
    month = "2",
    year = "2014"
}

@incollection{Torres-Vega17,
author = {Armando Martínez-Pérez and Gabino Torres-Vega},
title = {Exact Finite Differences for Quantum Mechanics},
booktitle = {Numerical Simulations in Engineering and Science},
publisher = {IntechOpen},
address = {Rijeka},
year = {2017},
editor = {Srinivas P. Rao},
chapter = {9},
doi = {10.5772/intechopen.71956},
url = {https://doi.org/10.5772/intechopen.71956}
}

@book{campos2019xft,
  title={The XFT quadrature in discrete Fourier analysis},
  author={Campos, Rafael G},
  year={2019},
  publisher={Springer},
  doi = {10.1007/978-3-030-13423-5}
}

@article{PhysRevLett.125.030503,
  title = {Reliability of Lattice Gauge Theories},
  author = {Halimeh, Jad C. and Hauke, Philipp},
  journal = {Phys. Rev. Lett.},
  volume = {125},
  issue = {3},
  pages = {030503},
  numpages = {6},
  year = {2020},
  month = {Jul},
  publisher = {American Physical Society},
  doi = {10.1103/PhysRevLett.125.030503},
  url = {https://link.aps.org/doi/10.1103/PhysRevLett.125.030503}
}

@article{Liu:2021tef,
    author = "Liu, Hanqing and Chandrasekharan, Shailesh",
    title = "{Qubit Regularization and Qubit Embedding Algebras}",
    eprint = "2112.02090",
    archivePrefix = "arXiv",
    primaryClass = "hep-lat",
    doi = "10.3390/sym14020305",
    journal = "Symmetry",
    volume = "14",
    number = "2",
    pages = "305",
    year = "2022"
}

@article{Wiese:2021djl,
    author = "Wiese, Uwe-Jens",
    title = "{From quantum link models to D-theory: a resource efficient framework for the quantum
 simulation and computation of gauge theories}",
    eprint = "2107.09335",
    archivePrefix = "arXiv",
    primaryClass = "hep-lat",
    doi = "10.1098/rsta.2021.0068",
    journal = "Phil. Trans. A. Math. Phys. Eng. Sci.",
    volume = "380",
    number = "2216",
    pages = "20210068",
    year = "2021"
}

@article{davoudi2022general,
  title={General quantum algorithms for Hamiltonian simulation with applications to a non-Abelian lattice gauge theory},
  author={Davoudi, Zohreh and Shaw, Alexander F and Stryker, Jesse R},
  journal={arXiv preprint arXiv:2212.14030},
  year={2022},
  doi = {10.48550/arXiv.2212.14030}
}

@article{garofalo2022defining,
  title={{Defining Canonical Momenta for Discretised $SU(2)$ Gauge Fields}},
  author={Garofalo, Marco and Hartung, Tobias and Jansen, Karl and Ostmeyer, Johann and Romiti, Simone and Urbach, Carsten},
  journal={arXiv preprint arXiv:2210.15547},
  year={2022},
  doi = {10.48550/arXiv.2210.15547}
}

@article{Hartung:2022hoz,
    author = "Hartung, Tobias and Jakobs, Timo and Jansen, Karl and Ostmeyer, Johann and Urbach, Carsten",
    title = "{Digitising SU(2) gauge fields and the freezing transition}",
    eprint = "2201.09625",
    archivePrefix = "arXiv",
    primaryClass = "hep-lat",
    doi = "10.1140/epjc/s10052-022-10192-5",
    journal = "Eur. Phys. J. C",
    volume = "82",
    number = "3",
    pages = "237",
    year = "2022"
}

@article{POLYCHRONAKOS2023116314,
title = {Composing arbitrarily many SU(N) fundamentals},
journal = {Nuclear Physics B},
volume = {994},
pages = {116314},
year = {2023},
issn = {0550-3213},
doi = {https://doi.org/10.1016/j.nuclphysb.2023.116314},
url = {https://www.sciencedirect.com/science/article/pii/S0550321323002432},
author = {Alexios P. Polychronakos and Konstantinos Sfetsos}
}

@article{Jakobs:2023lpp,
    author = "Jakobs, Timo and Garofalo, Marco and Hartung, Tobias and Jansen, Karl and Ostmeyer, Johann and Rolfes, Dominik and Romiti, Simone and Urbach, Carsten",
    title = "{Canonical momenta in digitized SU(2) lattice gauge theory: definition and free theory}",
    eprint = "2304.02322",
    archivePrefix = "arXiv",
    primaryClass = "hep-lat",
    doi = "10.1140/epjc/s10052-023-11829-9",
    journal = "Eur. Phys. J. C",
    volume = "83",
    number = "7",
    pages = "669",
    year = "2023"
}

@misc{garofaloLattice2023,
  author = {Garofalo, Marco and Hartung, Tobias and Jakobs, Timo and Jansen, Karl and Ostmeyer, Johann and Rolfes, Dominik and Romiti, Simone and Urbach, Carsten},
  title = {{Canonical Momenta in Digitized $SU(2)$ Lattice Gauge Theory}},
  url = {https://indico.fnal.gov/event/57249/contributions/270647/}
}

@misc{romitiLattice2023,
  author = {Romiti, Simone and Garofalo, Marco and Hartung, Tobias and Jakobs, Timo and Jansen, Karl and Ostmeyer, Johann and Rolfes, Dominik and Urbach, Carsten},
  title = {{Simulating the lattice $SU(2)$ Hamiltonian with discrete manifolds}},
  url = {https://indico.fnal.gov/event/57249/contributions/268337/}
}

@article{Zache:2023dko,
    author = "Zache, Torsten V. and Gonz\'alez-Cuadra, Daniel and Zoller, Peter",
    title = "{Quantum and Classical Spin-Network Algorithms for q-Deformed Kogut-Susskind Gauge Theories}",
    eprint = "2304.02527",
    archivePrefix = "arXiv",
    primaryClass = "quant-ph",
    doi = "10.1103/PhysRevLett.131.171902",
    journal = "Phys. Rev. Lett.",
    volume = "131",
    number = "17",
    pages = "171902",
    year = "2023"
}

@article{Bauer:2023jvw,
    author = "Bauer, Christian W. and D'Andrea, Irian and Freytsis, Marat and Grabowska, Dorota M.",
    title = "{A new basis for Hamiltonian SU(2) simulations}",
    eprint = "2307.11829",
    archivePrefix = "arXiv",
    primaryClass = "hep-ph",
    reportNumber = "IQuS@UW-21-058",
    month = "7",
    year = "2023",
    doi={10.48550/arXiv.2307.11829}
}

@misc{DJTPaperRepo,
    doi={10.5281/zenodo.10143896},
    author = {Romiti, Simone},
    title = {\text{su2\_DJT-v1.0}},
    url = {https://github.com/simone-romiti/DJT_paper-code},
    year = "2023"
}

@article{Garofalo:2023zkd,
    author = "Garofalo, Marco and Hartung, Tobias and Jakobs, Timo and Jansen, Karl and Ostmeyer, Johann and Rolfes, Dominik and Romiti, Simone and Urbach, Carsten",
    title = "{Testing the $\mathrm{SU}(2)$ lattice Hamiltonian built from $S_3$ partitionings}",
    eprint = "2311.15926",
    archivePrefix = "arXiv",
    primaryClass = "hep-lat",
    month = "11",
    year = "2023",
    doi={10.48550/arXiv.2311.15926}
}

\end{document}